\documentclass[journal]{IEEEtran}
\usepackage{graphicx}
\usepackage{amsthm}
\usepackage{epsfig}
\usepackage{latexsym}
\usepackage{amsfonts}
\usepackage{here}
\usepackage{rawfonts}
\usepackage[T1]{fontenc}
\usepackage{calc}
\usepackage{capitalgreekitalic}
\usepackage{url}
\usepackage{enumerate}
\usepackage{color}
\usepackage[tbtags]{amsmath}
\usepackage{amssymb}
\usepackage{upref}
\usepackage{epic,eepic}
\usepackage{soul}
\usepackage{times}
\usepackage{dsfont}
\usepackage{comment}
\usepackage{cite}
\usepackage{bbm} 
\usepackage{amsmath}
\usepackage{microtype} 
\usepackage{afterpage}
\usepackage{makecell}
\newtheorem{lemma}{\bf Lemma}

\usepackage{dsfont}
\newcounter{step}
\newlength{\totlinewidth}
\newenvironment{algorithm}{%
  \rule{\linewidth}{1pt}
  \begin{list}{}%
    {\usecounter{step}%
      \settowidth{\labelwidth}{\textbf{Step 2:}}%
      \setlength{\leftmargin}{\labelwidth}%
      \setlength{\topsep}{-2pt}%
      \addtolength{\leftmargin}{\labelsep}%
      \addtolength{\leftmargin}{2mm}%
      \setlength{\rightmargin}{2mm}%
      \setlength{\totlinewidth}{\linewidth}%
      \addtolength{\totlinewidth}{\leftmargin}%
      \addtolength{\totlinewidth}{\rightmargin}%
      \setlength{\parsep}{0mm}%
      \raggedright}}%
  {\end{list}%
  \rule{\linewidth}{1pt}}
\newcounter{substep}

  {\end{list}}

\newlength{\aligntop}
\newlength{\alignbot}
 \makeatletter
\renewenvironment{align}{%
  \vspace{\aligntop}
  \start@align\@ne\st@rredfalse\m@ne
}{%
  \math@cr \black@\totwidth@
  \egroup
  \ifingather@
    \restorealignstate@
    \egroup
    \nonumber
    \ifnum0=`{\fi\iffalse}\fi
  \else
    $$%
  \fi
  \ignorespacesafterend%
  \vspace{\alignbot}\par\noindent
} \makeatother

\IEEEoverridecommandlockouts

\newcommand{\nbl}[1]{#1}

\makeatletter
\newcommand{\nblcaption}[1]{%
  \let\nbl@old@makecaption\@makecaption
  \long\def\@makecaption##1##2{%
    \nbl@old@makecaption{\nbl{##1}}{\nbl{##2}}}%
  \caption{#1}%
  \let\@makecaption\nbl@old@makecaption
}
\newcommand{\nblbibkey}[1]{\expandafter\gdef\csname nblbib@#1\endcsname{}}
\nblbibkey{lyu2026quantizationawarecollaborativeinferencelarge}
\nblbibkey{JMLR:v21:20-081}
\nblbibkey{lowe2017multi}
\nblbibkey{tan1993multi}
\newcommand{\nbl@bibcolor}[1]{\color{black}}
\let\nbl@old@bibitem\@bibitem
\def\@bibitem#1{%
  \nbl@bibcolor{#1}%
  \nbl@old@bibitem{#1}%
  \nbl@bibcolor{#1}}
\let\nbl@old@lbibitem\@lbibitem
\def\@lbibitem[#1]#2{%
  \nbl@bibcolor{#2}%
  \nbl@old@lbibitem[#1]{#2}%
  \nbl@bibcolor{#2}}
\makeatother
\usepackage{algorithm}%,algpseudocode}
\usepackage{algorithmic}
\makeatletter
\newcommand\semihuge{\@setfontsize\semihuge{19.3}{25}}
\makeatother

\makeatletter
\newcommand\semismall{\@setfontsize\semihuge{12.4}{15}}
\makeatother

\renewcommand{\algorithmicrequire}{ \textbf{Input:}} %Use Input in the format of Algorithm
\renewcommand{\algorithmicensure}{ \textbf{Init:}} %UseOutput in the format of Algorithm
\usepackage{subfigure}
\begin{document}

\title{\Huge Optimization of Collaborative Semantic Communication Network Performance with Channel and Content Preference Feedback}

\author{\large{
Defeng Zhou, Dongyu Wei, \textit{Graduate Student Member, IEEE,} Siyao Li, \\ Mingzhe Chen, \textit{Senior Member, IEEE}

\thanks{Defeng Zhou, Dongyu Wei, and Mingzhe Chen are with the Department of
Electrical and Computer Engineering, University of Miami, Coral Gables, FL,
33146, USA (e-mail: \{defengzhou, dongyu.wei, mingzhe.chen\}@miami.edu)

Mingzhe Chen is also with the Frost Institute for Data Science and
Computing, University of Miami, Coral Gables, FL 33146 USA.

Siyao Li is with the Department of Electrical Engineering and Computer Science, Embry-Riddle Aeronautical University,
Daytona Beach, FL 32114, USA (email: lis14@erau.edu)}
\vspace{-0.9cm}} 
% \IEEEauthorblockA{
% \IEEEauthorrefmark{1}Department of Electrical and Computer Engineering, University of Miami, Coral Gables, FL, 33146, USA\\
% \IEEEauthorrefmark{2}Department of Industrial and Systems Engineering, University of Miami, Coral Gables, FL, 33146, USA\\
% \IEEEauthorrefmark{3}Frost Institute for Data Science and Computing, University of Miami, Coral Gables, FL, 33146, USA\\
% Emails:\{defengzhou, dongyuwei, yehu, mingzhechen\}@miami.edu}

%\vspace*{1em}\\ 
}%\vspace*{-2em}

%\thanks{Hanzhi Yu and Mingzhe Chen are with the Department of Electrical and Computer Engineering and Frost Institute for Data Science and Computing, University of Miami, Coral Gables, FL 33146 USA (Emails: \protect\url{{hanzhiyu, mingzhe.chen}@miami.edu)}.} 
%\thanks{Yuchen Liu is with the Department of Computer Science, North Carolina State University, Raleigh, NC 27695 USA (Email: \protect\url{yuchen.liu@ncsu.edu}).}
%R\thanks{This work was supported by the U.S. National Science Foundation under Grants CNS-2312139 and CNS-2312138. }
%
%

%
\maketitle
%
%%
%\vspace{0cm}
\begin{abstract}
\nbl{Semantic communications is attracting attention due to its ability to transmit data meaning rather than raw bits. However, existing semantic communication frameworks treat and transmit all image regions with equal importance, which is not practical for real-world applications which may prioritize different content in an image.
To address this issue, we propose a novel semantic communication framework that enables a transmitter to use limited channel and content feedback to prioritize the transmission of important image regions.
In particular, in the proposed framework, a base station (BS) divides each image into sub-images, extracts their semantic information, and transmits them to users according to their preferences.
The users will reconstruct the image based on the received sub-images and cooperatively decide when to send channel state information (CSI) or content-preference feedback under dynamic channels and limited resources.
We formulate an optimization problem to minimize the semantic-weighted mean square error between the original image and the regenerated image by optimizing sub-channel allocation, users' power allocation, and feedback selection. To address this problem, a value decomposition actor-critic (AC) with dynamic neighborhood construction (VDAC-DNC) scheme is proposed.
The proposed method combines AC with value decomposition networks to allow the BS to approximate discrete actions by a continuous action distribution, thus reducing the output dimension and improving training efficiency. The introduced DNC method further improves training efficiency by constructing a small discrete neighboring action space to search for an action with the maximum Q value, thus avoiding traversing the large discrete action space.
Simulation results show that the proposed VDAC-DNC scheme can improve the performance by up to 5.04\% and 18.55\% compared to the standard multi-agent QAC method and the proposed method without feedback transmission. 
% Meanwhile, the proposed method with DNC achieves a 2.9x-5.3x speedup over the proposed method without DNC.
}
\end{abstract}


\section{Introduction}\label{se:intro}
% semantic communication (image) importance
% challenge
% Modern communication systems continue to push toward the Shannon capacity limit, which defines the maximum reliable data rate of a physical channel~\cite{9562559}. At the same time, 
The integration of communication and artificial intelligence (AI) is moving wireless communications beyond bit-level transmission toward meaning-level exchange~\cite{10554663, 10646587}. In this paradigm, semantic communication (SC) has emerged as a key technique, enabling transmitters to extract and transmit the ``meaning'' of data rather than raw bits, thereby allowing receivers to reconstruct the data or infer task outputs (e.g., classification labels~\cite{sagduyu2023task}, answers~\cite{xie2021task}) with significant reduced latency and resource consumption~\cite{10855638}. 
%Given its focus on data meaning instead of data bits, SC shifts the design from bit data rate maximization to data meaning maximization thus reducing data transmission latency and improve transmission efficiency under tight resource constraints~\cite{10855638}.
However, deploying SC techniques over wireless networks faces several challenges including: 1) data meaning extraction and representation, 2) semantic information allocation in dynamic wireless networks, and 3) content feedback-assisted semantic communication method design.
% 语义提取
% dynamic
% content（segmentation），preference（feedback）
\subsection{Related Works}
%
% Recently, several works in~\cite{zhang2023optimization, 9746335,10015684,10510413, 10500305, ding2024adaptive, WU2024519} have studied the optimization of image transmission in SC networks.
Recently, several works in~\cite{9746335,10388062, 10960269, 10015684, 10094735} have studied the optimization of image transmission in different wireless networks.
In~\cite{9746335}, the authors designed a deep joint source-channel coding (JSCC) scheme that adapts its compression rate to channel conditions by training a policy network to select features from the image that need to transmit.
The author in~\cite{10388062} investigated an end-to-end JSCC design by unifying the coding rate reduction maximization and the mean square error (MSE) minimization for image recovery and classification simultaneously.
The author in~\cite{10960269} proposed probabilistic graph-based method for image transmission, enabling semantic information to be compressed and transmitted more accurately while reducing communication energy consumption.
The authors in~\cite{10015684} introduced a simple code-mask–based JSCC that adjusts the compression rate to meet a target image quality (e.g., peak signal to noise ratio) under different wireless networks.
%In~\cite{10646587}, the authors used AIGC to 
In~\cite{10094735}, the authors employ Swin Transformer to extract semantic information of image and scaled the latent representation according to \nbl{channel state information (CSI)}, enhancing the ability of the model to deal with various channel conditions.
% In~\cite{li2025star}, the author designed a lightweight deep JSCC framework and introduced an improved channel state adaptive module to adapt to different channel conditions. The module is based on an attention mechanism and incorporates SNR information as a guidance factor to refine the semantic encoder-decoder, thereby achieving channel state adaption.
However, the works in~\cite{9746335,10388062, 10960269, 10015684, 10094735} treat all image regions with equal importance, which is not practical for real-world users who prioritize different content (e.g., a ``car'' content may be more semantically relevant than the ``sky'' in a traffic scenario). Without accounting for content preference, transmitters may waste bandwidth on irrelevant background features while important regions suffer from distortion.
%did not consider content importance. They only focus on the overall coding rate and do not apply differential compression based on content, making it impossible to allocate resources according to the importance of the content.
%没有根据内容进行不同的压缩，这样就无法根据content的重要程度进行resource allocation？

% When channel SNR is low and channel decoding is not able to guarantee the zero bit error rate, the quality of the recovering result will be dramatically reduced.
% This problem widely exists in conventional wireless communication netowe
% In~\cite{zhang2023optimization}, the authors extracted the semantic information from images and modeled them as scene graphs that captures the objects and their relationships in the original image.

%说是semantic，但没有明显的关注内容？
%没有关注dynamic

% 关注压缩与环境变化,resource allocation 没有强调的content，而是只关注整体的压缩率
% 只有CSI feedback，可以结合上面一起说
%还有一些文献研究关注图片中物体的重要性

To address this problem and further improve transmission reliability, feedback mechanisms are essential.
The benefit of CSI feedback in downlink transmission has also been studied in capacity region~\cite{Li2019LPEBC-capacity} and stability region~\cite{Li2019LPEBC-stability}.
The authors in~\cite{10510413} and~\cite{10500305} integrated channel feedback into semantic image transmission. In particular, the work in~\cite{10510413} used CSI feedback to adjust the compression rate of each image.
%analyzed semantic-level distortion with feedback, while a transformer-aided 
The authors in~\cite{10500305} sent the decoder's current estimation regarding to the source image to the transmitter so as to update the encoder and mitigate the impact of channel noise on end-to-end reconstruction quality.
%Some literature studies have also examined the importance of objects within images.
The authors in~\cite{ding2024adaptive} simply prioritized the foreground of the image and transmitted it while masking less informative background.
The authors in~\cite{WU2024519} divided image regions into two parts: regions of interest (ROI) and regions of non-interest (RONI), and designed an SC neural network that can preserve more information to transmit the ROI regions with high quality while transmitting RONI with data compression.
\nbl{The authors in~\cite{lyu2026quantizationawarecollaborativeinferencelarge} considered a collaborative inference between the transmitter and receiver, and investigated the use of inference results of the receiver to optimize quantization bit-width and computation frequencies of the transmitter, so as to improve inference efficiency.}
% They further jointly optimized the on-agent quantization bit-width and the computation frequencies of the agent and server under delay and energy constraints, thus enhancing collaborative inference efficiency.
% \nbl{The authors in~\cite{lyu2026quantizationawarecollaborativeinferencelarge} seperated the inference process into two parts: the local inference at the receiver and the global inference at the transmitter. They transmitted the inference result of the receiver to the transmitter as feedback, thus the transmitter can know the partial information of the inference to improve the inference efficiency.}
% enabled high quality transmission of regions of interest (ROI) by using more complicated semantic communication networks, while using less bandwidth to transmit regions of non-interest (RONI) with data compression.
However, current works in~\cite{10510413, 10500305, ding2024adaptive, WU2024519,lyu2026quantizationawarecollaborativeinferencelarge} did not consider how dynamic wireless channels and user data preference affect semantic image transmission.
%the image transmission over a network where wireless channel are dynamic, and data preference of different users simultaneously. 
%image transmissions over a network where wireless channel are dynamic, and transmitters do not have perfect CSI and data preference information. 
Specifically, with imperfect CSI, the transmitter may overestimate a poor channel thus resulting in packet loss, or underestimate a good channel to waste bandwidth.
Meanwhile, without data preference information, important image components may be scheduled on poor channels while less important regions are protected, which may increase semantic distortion and reduce task accuracy.
%These effects lead to unstable quality and higher latency, decreasing the performance of SC severely. 
Hence, when considering both dynamic wireless channels and user data preference, users must trade off between sending CSI to inform the \nbl{base station (BS)} about the status of wireless channels as well as sending content preference information to request important content. 
% \nrd{From an information-theoretic perspective, the benefit of CSI feedback in downlink broadcast settings has also been studied via the layered packet erasure channels, including characterizations of its capacity region~\cite{Li2019LPEBC-capacity} and stability region~\cite{Li2019LPEBC-stability}. }
% 在没有perfect CSI的情况下，
% either optimize only the overall image quality without modeling what is important to the user, or rely on pre-defined importance that is not decided by the user.
% Meanwhile, there is limited research focus on the user preference feedback and explicit meaning of image. Most of current methods do not concern the meaning of the image itself or only consider its implicit expression.
% 并不关注图像本身的意思或者隐式的表达，而我们的方法通过图像分割，以此显示的关注到图像的语义，
% 第一个preference的importance只取决于图像本身的特点，而不是用户的关注点，而不是随着用户动态改变，
% 第二个工作只简单把所有内容分为两类，不够flexible dynamic semantic 
% recently, semantic communication image transmission 优化的方向
% 考虑到很少有人把信道分配与图像
% 用户的喜好

%RL,动作空间大 我们引入DNC

A number of existing works have investigated reinforcement learning (RL) for optimizing semantic communication performance. 
%the use of RL for optimizing the performance of semantic communication.
% The authors in \cite{11025994} proposed an adversarial RL, which consists of an information RL and a question RL to transmit semantic information to users with heterogeneous knowledge.
The authors in~\cite{10681776} proposed a deep reinforcement learning driven joint position and
power optimization algorithm to maximize the semantic data transmission throughput.
In \cite{9832831}, an attention-based RL algorithm was developed to analyze the relationship between the original data and its semantic information.
% \nbl{The authors in \cite{10993496} designed a hierarchical RL-based optimization framework to joint optimize the mixed-integer nonlinear programming (MINLP) problem.}
However, these works~\cite{10681776, 9832831} do not consider the cooperation among different agents. 
Consequently, each agent's performance will be affected by the actions of other agents, thus reducing the transmission performance achieved by RL.
To address the performance degradation caused by the lack of collaboration among different agents, the authors in \cite{10149174} proposed a value decomposition-based deep Q-learning scheme multi-agent RL (MARL) to minimize the average transmission latency.
The authors in \cite{zhang2025semanticawareresourcemanagementcv2x} utilized a multi-agent deep deterministic policy gradient (MADDPG)-based MARL to dynamically allocate channels and power for transmission.
Nonetheless, these RL algorithms in~\cite{10149174, zhang2025semanticawareresourcemanagementcv2x} struggle to scale to large discrete action spaces since they must evaluate and explore a huge number of actions~\cite{10.1007/s10994-021-05961-4}.

% contribution
% 
% 现有的模型基本不考虑用户的preference根据用户的preference，决定sub-image的importance，在资源有限的情况下，进行合理的资源分配，确保其传输的稳定性
% 同时适合可变的CSI，preference
% 引入了一种高效的RL 算法
\subsection{Contributions}
The main contribution of our work is a novel channel and content preference feedback enabled SC framework that enables multiple users to use limited channel and content feedback for efficient and accurate image data transmission.
The key contributions include:
\begin{itemize}
    % \item We consider a novel channel and content preference feedback enabled SC framework in which a BS transmits the images to a set of users while users cooperatively transmit the feedback to the BS for further efficient and accurate image data transmission.
    % In particular, the BS first divides the images into several sub-images according to images and allocates these sub-images to different sub-channels and sends them to each user using SC techniques. After receiving the sub-images, users may transmit the CSI feedback (i.e., the representation of successfully received sub-images) or content preference feedback (i.e., the importance of different sub-image categories) to the BS, which will use this information to optimize sub-channel allocation.
    % Considering limited power for each user and the limited uplink bandwidth, users must cooperate to transmit the feedback under the required latency.
    % To minimize the MSE between the original images and the images regenerated by users, we formulate an optimization problem aiming to minimize MSE while satisfying the feedback transmission latency requirements of users by optimizing the BS's sub-image allocation, user power allocation and feedback selection.
    \item We consider a novel channel and content preference feedback enabled SC framework in which a BS transmits the images to a set of users while users cooperatively transmit the feedback to the BS for further efficient and accurate image data transmission.
    In particular, the BS first divides the images into several sub-images according to image content and allocates these sub-images to different sub-channels and sends them to each user using SC techniques. After receiving the sub-images, users may transmit the CSI feedback or content preference feedback to assist the BS in subsequent sub-image transmission.
    Considering limited power for each user and the limited uplink bandwidth, users must cooperate to transmit the feedback under the required latency.
    To minimize the semantic-weighted MSE between the original images and the images regenerated by users, we formulate an optimization problem aiming to minimize MSE while satisfying the feedback transmission latency requirements of users by optimizing the BS's sub-image allocation, user power allocation and feedback selection.
    % \item % user会竞争，
    % Considering limited power for each user and the limited uplink bandwidth, users must cooperate to transmit the feedback under the required latency. 
    % On the one hand, the earlier the user transmits the feedback to the BS, the BS will have more complete information of channel and sub-image to ensure the accuracy of image transmission. On the other hand, users must cooperate to transmit the feedback under the required latency due to the limited power and bandwidth of the uplink channel.  
    % To minimize the MSE between the original images and the images regenerated by users, we formulate an optimization problem aiming to minimize MSE while satisfying the feedback transmission latency requirements of users by optimizing the BS's sub-image allocation, user power allocation and feedback selection.
    % We formulate this problem as an optimization problem aiming to minimize the mean square error (MSE) between the original images and the images regenerated by users.

    \item To solve the optimization problem, a value decomposition-based actor-critic with dynamic neighborhood construction (VDAC-DNC) method is proposed, which enables the BS to allocate the sub-channel and users feedback selection and power allocation cooperatively. 
    Different from the standard value decomposition-based networks (VDN)~\cite{10.5555/3237383.3238080},
    the proposed method combines actor-critic with VDN, which allows the BS to approximate the discrete actions by a continuous action distribution, thus reducing the output dimension and improving RL training efficiency.
    % In addition, the proposed method is more
    % effective in finding better actions in huge discrete action space
    % since the designed method approximates the discrete action
    % space using a continuous action and selects an action with
    % the maximum $Q$ value from the constructed neighboring action space.
    \item To further improve the training efficiency, the introduced DNC method finds a better discrete action from the continuous action per RL iteration. In particular, the DNC method constructs a small discrete neighboring action space to search for an action with the maximum $Q$ value, thus avoiding traversing the large discrete action space.
In particular, the DNC uses a continuous action produced by a continuous policy, which is easier to optimize than a complex discrete policy over large discrete action spaces, to construct a small discrete neighboring action space. 
Then, DNC searches for an action with the maximum $Q$ value from the small neighboring action space, thus avoiding traversing the large discrete action space.
\end{itemize}
Simulation results show that the proposed VDAC-DNC scheme can improve the performance by up to 
5.04\%, 12.43\%, and 18.55\% compared to the standard multi-agent QAC (MAQAC) method~\cite{10.5555/3295222.3295385}, independent Q method, and the proposed method without CSI and content preference feedback information. Meanwhile, the proposed method with DNC achieves a 2.9x-5.3x speedup over the proposed method without DNC while maintaining performance.

The rest of this paper is organized as follows.
The proposed channel and content preference feedback enabled SC framework and problem formulation are described in Section~\ref{se:system}.
The proposed VDAC-DNC scheme for optimizing sub-image allocation, power allocation and feedback selection is introduced in Section~\ref{se:algorithm}.
In Section~\ref{se:analy}, we analyze the convergence, implementation and complexity of the proposed scheme.
%We further analyze the information-theoretic bounds on the semantic reconstruction error under the proposed framework in Section~\ref{se:theory}.
Numerical simulation results are shown and analyzed in Section~\ref{se:simulation}.
Finally, Section~\ref{se:conclusion} concludes this paper.
\section{System Model and Problem Formulation}\label{se:system}
\begin{figure}
    \centering
    \includegraphics[width=1\linewidth]{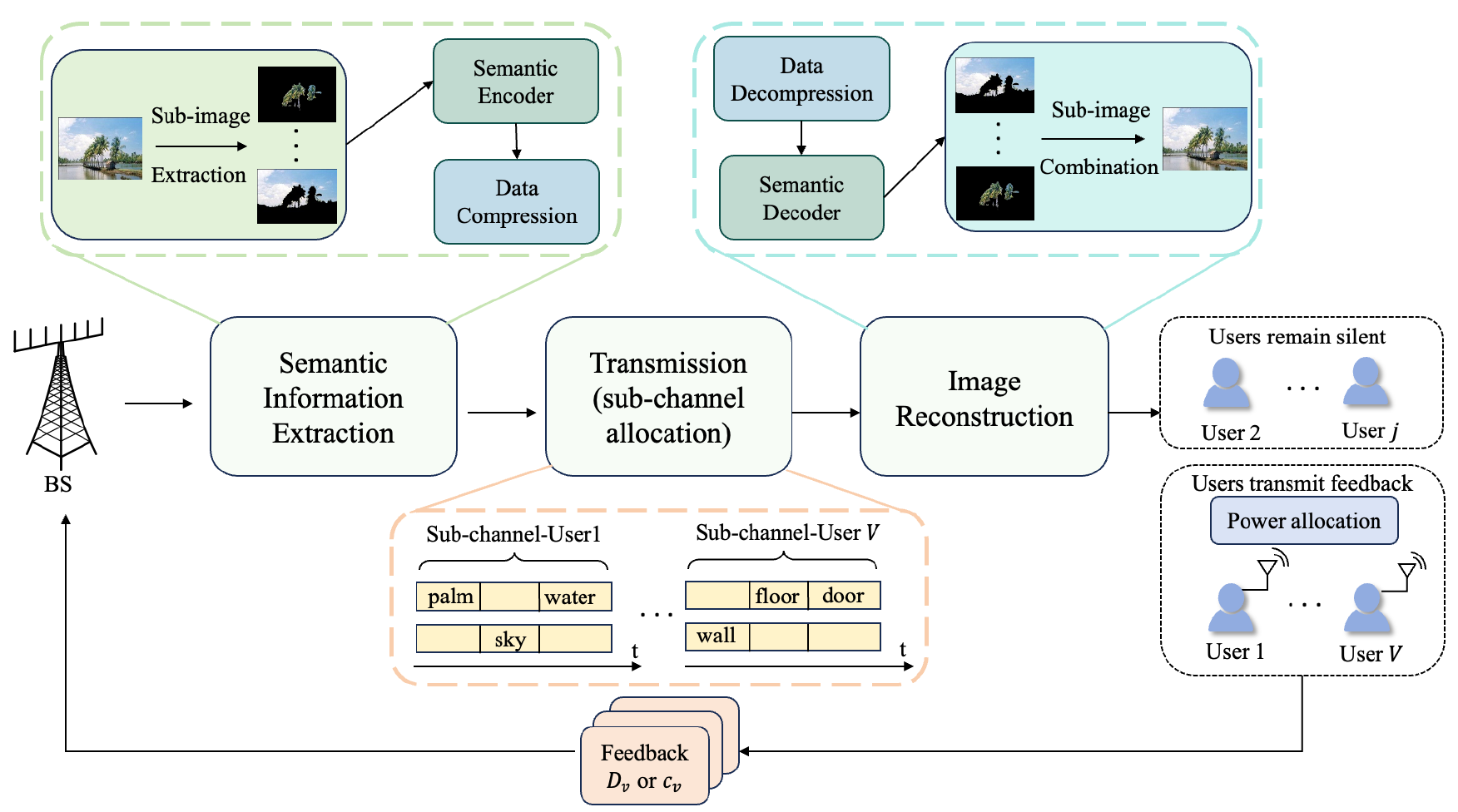}
    \nblcaption{Channel and preference feedback enabled semantic communication framework}
    \label{system}
\end{figure}

We consider a channel and preference feedback enabled SC system where a BS sends images to a set $\mathcal{V}$ of $v$ users using SC techniques, as illustrated in Fig.~\ref{system}. 
In particular, to send an image to a user, the BS
divides the image into several sub-images according to image content.
For example, the BS will extract "palm", "sky" and "water" from the original image and generate three sub-images.
Then, the BS uses a variational autoencoder (VAE) based semantic encoder~\cite{Rombach_2022_CVPR} to encode sub-images according to the downstream applications (e.g., image classification or object detection) of the receiver.
Hereinafter, we define the output of the semantic encoder as the semantic information of the sub-image.
%$I_k^j$ as the $j$-th class of $k$-th image semantic information.
The BS will allocate sub-image semantic information to different sub-channels and send them to the receiver over time-varying wireless channels.
At the receiver, a semantic decoder is used to reconstruct sub-images. Then, the receiver merges the sub-images to regenerate the original image.
Based on the regenerated image, the receiver will send CSI (i.e., the representation of successfully received sub-images) or content preference feedback (i.e., the importance of different sub-image categories) to the BS such that the BS can optimize sub-image allocation and feedback selection to improve semantic transmission performance.   
%the receiver can send two types of feedback to the transmitter.
% One is the channel state information (CSI) feedback. It mainly reports how many packets were successfully received, enabling the transmitter to adjust and further improve the quality of the reconstructed image.
% The other is the downstream task feedback. 
% It's about matching the transmission focus of the image to the user's needs.
% This directs the transmitter to prioritize some parts of the image according to the user's specific application or task requirements.
%先按照图像占比传输，若是用户不满意，则给出一个主体类别作为最重点传输！
%交通数据集：有的用户需要traffic light，有的用户侧重car，有的用户侧重行人
\nbl{The designed framework can be used for applications that need to transmit images such as autonomous driving, traffic monitoring, remote inspection, and sensing.}
Next, we first introduce the methods used for semantic information extraction. Then we introduce the semantic information transmission model and image reconstruction.
Finally, we formulate the optimization problem.
% it extracts semantic information from each category and transmits the extracted semantic information to the users according to their requirements.
% In particular, the considered semantic communication process consist of four phases: a) image data preprocessing, b) semantic information extraction, c) semantic-oriented resource allocation and  d) original image data recovery.

\subsection{Semantic Information Extraction}
%总分结构！
% segmentation
% mask
Semantic information extraction includes three key steps: 1) sub-image extraction, 2) semantic information generation, and 3) sub-image compression, which will be introduced in detail as follows.
\subsubsection{Sub-image Extraction}
We use the method in~\cite{wang2023internimageexploringlargescalevision} to divide an image into several sub-images according to image objects.
% InternImage is a new large-scale CNN-based foundation model, achieves performance comparable to or better than state-of-the-art ViTs on challenging benchmarks. 
Compared to other sub-image extraction methods (e.g., sliding-window convolutional neural network (CNN)~\cite{9880273} and fixed-grid vision transformer (ViT)~\cite{9710580}), the method from~\cite{wang2023internimageexploringlargescalevision} can identify important and meaningful pixels, achieving large receptive fields for sub-image extraction.
\nbl{Here, we can use any other semantic segmentation methods, which will not affect our designed semantic communication framework.}
% InternImage consists of a convolutional stem, multiple hierarchical stages of DCNv3-based deformable convolution blocks (each block including layer normalization and a feed-forward network), and a final classification head.
% Assuming the set of images for transmission is $\mathcal{I} = \{\boldsymbol{I}_1, \boldsymbol{I}_2, \dots, \boldsymbol{I}_K\}$.
Let $\boldsymbol{I}_{k, v} \in \mathbb{R}^{H \times W \times C}$ be an original image $k$ that the BS needs to transmit to a user $v$ with $H$, $W$ and $C$ being the height, width, and the number of color channels of the image. 
%the original images that are going to be transmitted are denoted as $\mathcal{I} = \{ \boldsymbol{I_1, \boldsymbol{I}_2}, \cdots, \boldsymbol{I}_K\},$
%The $k$-th image's segmention result is obtained based on the internimage model function $f_{\text{s}}$, denoted as follows: 
Then, the set of sub-images that are extracted from image $\boldsymbol{I}_{k,v}$ using the model~\cite{wang2023internimageexploringlargescalevision} is $\mathcal{M}_{k,v} = \{\boldsymbol{m}_{k,v}^1, \boldsymbol{m}_{k,v}^2, \dots, \boldsymbol{m}_{k,v}^{J_{k,v}}\}$ where $J_{k,v}$ is the number of sub-images in image $\boldsymbol{I}_{k,v}$ which varies according to the image $\boldsymbol{I}_{k,v}$, $\boldsymbol{m}^{j}_{k,v}$ is sub-image $j$ in image $\boldsymbol{I}_{k,v}$.
%Note that the value of $J$ is determined by $\boldsymbol{I}_{k}$.
% \begin{equation}
      
% \end{equation}

\subsubsection{Semantic Information Generation}
A VAE is used to generate semantic information. 
Although the size of the VAE input is not required to be fixed, the height and width of the input must be divisible by the VAE's downsampling factor (typically 8 or 16). However, the size of the sub-images is different which depends on the objects in the sub-images. For simplicity, we unify the size of the sub-images such that we can use one VAE model to process various sub-images and generate semantic information. 
To this end, we use a padding method that changes the size of the sub-image to the size of the original image by adding zeros around each sub-image.
% We denote $\hat{\boldsymbol{m}}^j_k$ as the sub-image $j$ with the same spatial dimensions as the original $\boldsymbol{I}_{k}$ after padding. 
%in which only pixels belonging to class $j$ are preserved and all other pixels are set to black.
% For example, we assume the $\hat{\boldsymbol{m}}^j_k$ represents the 'dog' and has the equal size to the original image, however, it only contains the pixels that are classified into the 'dog' and all other pixels are set to black.
% Next, we use the encoder, denoted as $f^{\textrm{EN}}(\cdot; \boldsymbol{\theta})$, where $\boldsymbol{\theta}$ represents the model parameters, to encode the selected sub-images into semantic information.
% The encoded symbol sequence of $j$-th class in $k$-th image $\boldsymbol{x}_k^j$ can be expressed as:
Then, the semantic information of padded sub-image $\widetilde{\boldsymbol{m}}^j_{k,v}$ is
\begin{equation}
    \boldsymbol{x}_{k,v}^j = f_{\boldsymbol{\theta}}^{\textrm{E}}(\widetilde{\boldsymbol{m}}^j_{k,v}),
\end{equation}
where $f_{\boldsymbol{\theta}}^{\textrm{E}}(\cdot)$ is the semantic encoder used to generate the semantic information of sub-image with $\boldsymbol{\theta}$ being the encoder parameters.
Here, the encoder can be applied for all users.
\subsubsection{Sub-image Compression}
% importance order
% Due to the transmission delay constraint and limited bandwidth, the BS must determine partial sub-images (i.e., a subset of original image) to be transmitted. 
% Hence we introduce the importance order, which indicates the importance of each sub-image. Then we can select sub-images which are the most necessary for image recovery and user's comprehension to transmit according to the importance order. Note that the importance order can be changed by the user's task feedback.
% Hence we assign each sub-image an importance score reflecting its contribution to overall image reconstruction and users' favor. We then select the highest-scoring sub-images for transmission. Note that this importance ordering can be dynamically updated in response to users' task feedbacks.
% compress
% encoder 
% 先讲结果再讲原因，还是先讲原因，再讲结果
% The sub-image can be compressed to save the transmission resource since large continuous areas with the same values are introduced in semantic information generation processing. 
% The feature generated from the filling black regions tend to be the same in value since convolution operations use shared kernels. 
% Since we use a padding method to increase the size of sub-images to be equal to the size of the original image, a large number of zero-valued pixels still exists in the sub-images.
Since $\widetilde{\boldsymbol{m}}^j_{k,v}$ was processed by a padding method, a large number of zero-valued pixels exist. 
Hence, we need to use a compression method to further reduce the redundancy in the sparse semantic information $\boldsymbol{x}_{k,v}^j$.
% If these repeated values are considered as zeros, 
% the feature matrix can be viewed as relatively sparse.
%If these repeated values are considered as , 
%In order to reduce the resource used to transmit the feature sparse matrix, 
Here, we propose to use the algorithm in~\cite{WU2024519} to compress the semantic information $\boldsymbol{x}_{k,v}^j$ of each padded sub-image $\widetilde{\boldsymbol{m}}_{k,v}^j$.
Let $f^{\textrm{C}}(\cdot)$ be the data compression function.
The compressed semantic information is $\hat{\boldsymbol{x}}_{k,v}^j = f^{\textrm{C}}(\boldsymbol{x}_{k,v}^j)$.
Here, we can also consider other compression methods, which do not affect the designed semantic data transmission methods. 

\subsection{Transmission Model}
Next, we introduce 1) the downlink transmission of sub-images from the BS to users, and 2) the uplink transmission of feedback from users to the BS.
\subsubsection{Downlink Sub-image Transmission}
% A time-slotted frame structure is employed, with the set of
% time slots denoted as $\mathcal{T} = \{1, 2, \dots  , T\}$.
We assume that the BS has a total of $N V$ channels and each user can use $N$ sub-channels for sub-image transmission. 
Let $\mathcal{L}_v$ be a set of sub-channels that the BS uses to serve user $v$.
% The BS can use a set $\mathcal{N}$ of $N$ sub-channels to transmit semantic information of sub-images for each user. 
% Let $\mathcal{L}_v$ represents a set of sub-channels that are used to serve user $v$.
These sub-channels have different channel conditions, and hence, they have different data rates and packet error rates.
%where $K$ is the number of original images that are going to be transmitted. 
% 为了简便，我们考虑每传完一张图片，再开始传下一张图片。为了防止一张图片占用过多时间，我们规定一张图片最多只能占用T个time slot，一旦超过时间，将强制开始传输下一张图片
% For simplicity, we assume that the transmission of each image is completed before the transmission of the next one begins. 
% To prevent a single image from occupying excessive time, we impose a constraint that each image can use at most $T$ time slots. Once this limit is exceeded, the system is forced to switch to transmitting the next image.
% We assume that the channel state is static over $T$ time slots.
% When the channel state changes, the number of packets that the receiver can successfully receive is changed.
% It also denotes how many layers of packets can be successfully received in our considered time-varying wireless channels.
% In particular, under the LPEC model, at a channel use the transmitter emits a vector of packet layers, and the receiver successfully recovers the first $S$ layers. 
% When a layer is erased, all higher-indexed layers with lower SNR are also erased, this correlated erasure berhavior captures the high-SNR dynamics of the fading channel.
% In our model, we consider a set $\mathcal{Q}$ of $Q$ downlink orthogonal resource blocks (RBs) can be allocated to serve the user.  Here, $\alpha_{q}^t$ implies that RB $q$ in the $t$-th time slot is allocated; otherwise, we have $\alpha_{q}^t = 0$. 
The downlink channel capacity of the BS transmitting semantic information to user $v$ over sub-channel $n$ at time slot $t$ is
% \begin{equation}
%     R_{n,v}^{\textrm{D}}(t) = 
%     \begin{cases}
%         \log_2\!\left(1+ \dfrac{P^{\textrm{D}} \boldsymbol{h}_{n,v}^{\textrm{D}}(t)}{I_{n,v}^{\textrm{D}} +  N_0}\right),   \textrm{w/ updated CSI} , %\nbl{ \triangleq R_{n,v}^{\textrm{D, CSI}}(t)}, 
%         \\[4pt]
%         \log_2\!\left(1 -\dfrac{P^{\textrm{D}} d^{-2}_v}{I_{n,v}^{\textrm{D}} +  N_0} \ln (1-\xi)\right),    \textrm{w/o updated CSI}.%\nbl{ \triangleq R_{n,v}^{\textrm{D, noCSI}}(t)}, 
%     \end{cases}
%     \label{downrate}
% \end{equation}
\begin{equation}
\begin{aligned}
R_{n,v}^{\textrm{D}}(t)=
\begin{cases}
\log_2\!\left(1+\dfrac{P^{\textrm{D}}\boldsymbol{h}_{n,v}^{\textrm{D}}(t)}{I_{n,v}^{\textrm{D}}+N_0}\right),~\text{with updated CSI},\\[2pt]
\log_2\!\left(1-\dfrac{P^{\textrm{D}}d_v^{-2}}{I_{n,v}^{\textrm{D}}+N_0}\ln(1-\xi)\right),~\text{otherwise},
\end{cases}
\end{aligned}
\label{downrate}
\end{equation}
% \begin{equation}
%     R_{n}^{\textrm{D}}(t) = 
%     \begin{cases}
%         \log_2\!\left(1+ \dfrac{P^{\textrm{D}} \boldsymbol{h}_{n}^{\textrm{D}}(t)}{I_{n}^{\textrm{D}} +  N_0}\right), & \textrm{with CSI}, %\nbl{ \triangleq R_{n,v}^{\textrm{D, CSI}}(t)}, 
%         \\[4pt]
%         \log_2\!\left(1 -\dfrac{P^{\textrm{D}} d^{-2}}{I_{n}^{\textrm{D}} +  N_0} \ln (1-\xi)\right), &  \textrm{without CSI}.%\nbl{ \triangleq R_{n,v}^{\textrm{D, noCSI}}(t)}, 
%     \end{cases}
%     \label{downrate}
% \end{equation}
where $P^{\textrm{D}}$ is the transmit power of the BS, $I_{n,v}^{\textrm{D}}$ is the interference caused by BSs that are located in other service areas and use sub-channel $n$, $N_0$ is the noise power spectral density, and $\boldsymbol{h}_{n,v}^{\textrm{D}}(t) = \gamma_n^{\textrm{D}} d_v^{-2}$ is the channel gain of sub-channel $n$ in time slot $t$ between the BS and user $v$ with $\gamma_n^{\textrm{D}} \sim \exp(\lambda =1)$ being the Rayleigh fading parameters of sub-channel $n$ and $d_v$ being the distance between the BS and user $v$, $\xi$ is the outage probability.
From~\eqref{downrate}, we see that when the BS does not have the updated CSI, the downlink transmission rate follows the standard outage-capacity formulation~\cite{tse2005fundamentals}.
\nbl{In particular, the BS does not have updated CSI in two cases: 1) the user does not transmit CSI feedback, or 2) the user transmitted a CSI feedback but the receiver did not receive it within the feedback delay constraint. }
% In these cases, the BS cannot get the instantaneous CSI and therefore cannot adapt the downlink transmission
% rate to the instantaneous channel realization and therefore uses the statistical channel distribution
% and the target outage probability to determine the no-CSI transmission rate.
% When the BS does not have access to instantaneous downlink CSI of user $v$ on sub-channel 
% $n$, the transmit rate cannot be adapted to the fading realization $\boldsymbol{h}_{n,v}^{\textrm{D}}(t)$. In this case, following the standard outage-capacity formulation for slow Rayleigh fading channels~\cite{tse2005fundamentals}, the constant downlink transmission rate $R_0^{\textrm{D}}$ is
% \begin{equation}
%     R_{n,v}^{\textrm{D}}(t) =   \log_2 (1 -\frac{P^{\textrm{D}} d^{-2}_v}{I_{n,v}^{\textrm{D}} +  N_0} \ln (1-\xi)), \textrm{without CSI}
% \end{equation}
% Here, $\alpha_{kj}^n(t)$ is the allocation indicator and is defined as
% \begin{equation}
%     \alpha_{kj}^n(t) = \left\{
%     \begin{aligned}
%         & \alpha^n_{kj}, &  (k-1)T \le t \le kT, \\
%         & 0, & \textrm{else},
%     \end{aligned}
%     \right.
% \end{equation}
% where $\alpha_{kj}^n \in \{0,1\}$ is a binary indicator.

\subsubsection{Uplink Feedback Transmission}
Here, we assume users can obtain the perfect CSI via pilot-assisted channel estimation.
Thus, the data rate of user $v$ transmitting feedback to the BS at time slot $t$ is
\begin{equation}
\label{eq:feedback-rate}
    R_{v}^\textrm{U} (t) = \frac{B}{g(t)} \log_2 (1+ \frac{g(t)p_v^\textrm{U}(t) \gamma^\textrm{U} d_v^{-2}}{ B N_0}), 
\end{equation}
where $p_v^\textrm{U}(t)$ is the transmit power of user $v$ at time slot $t$, $g(t)$ is the number of users that want to transmit feedback to the BS, $B$ is the total uplink bandwidth, $N_0$ is the noise power spectral density, and $\gamma^{\textrm{U}}$ is the uplink Rayleigh fading parameter. 
% Thus, the latency of user $v$ transmitting selected feedback to the BS is
% \begin{equation}
%     \tau_v(\beta_{\textrm{D},v}(t), \beta_{\textrm{C},v}(t)) = \frac{\beta_{\textrm{D},v}(t)Z( \boldsymbol{D}_v(t)) + \beta_{\textrm{C},v}(t)Z(\boldsymbol{c}_v(t))}{R^{\textrm{U}}_v(t) },
% \end{equation}
% where $Z(\boldsymbol{D}_v(t) )$ and $Z(\boldsymbol{c}_v(t))$ are the data sizes of content feedback and CSI feedback. These feedbacks will be introduce in Section~\ref{sec:fb}.

\subsection{Image Reconstruction}
%decoder
Image reconstruction includes the following three key steps: 1) Data decompression, 2) sub-image regeneration, and 3) sub-image combination, which will be introduced in detail as follows.
%image combination
\subsubsection{Data Decompression}
% As mentioned above, the data was compressed to save the bandwith resource.
% Since the size of 
Once the compressed data is received, the user restores it back to its original form for sub-image generation.
Here, we apply the restoration algorithm in~\cite{WU2024519}. 
% Let $f^{\textrm{R}}$ be the data restoration function. 
% Thus, the restored semantic information of sub-image $j$ is $\boldsymbol{x}^{j}_k = f^{\textrm{R}}(\hat{\boldsymbol{x}}^{j}_k)$, where $\hat{\boldsymbol{x}}^{j}_k$ is the successfully received compressed semantic information. 

\subsubsection{Sub-image Regeneration}
A semantic decoder of each user $v$ is used to regenerate the corresponding sub-image $\boldsymbol{s}_{k,v}^j$.
Since user $v$ may not be able to receive all the sub-images due to channel packet errors, we define $\boldsymbol{\eta}_{k,v}(t) = [\eta_{k1,v}(t), \eta_{k2,v}(t), \dots, \eta_{kJ,v}(t)]$ as a vector to represent whether the sub-images are received by user $v$.
In particular, $\eta_{kj,v} = 1$ implies that semantic information $\hat{\boldsymbol{x}}_{k,v}^j$ is received by user $v$, and $\eta_{kj,v} = 0$, otherwise.
The sub-channel allocation vector of compressed semantic information for user $v$ is $\boldsymbol{\alpha}_{k,v}(t) = \left[ \boldsymbol{\alpha}_{k1,v}(t), \dots ,\boldsymbol{\alpha}_{kj,v}(t), \dots, \boldsymbol{\alpha}_{kJ,v}(t) \right]$ with $\boldsymbol{\alpha}_{kj,v}(t) = [\alpha_{kj,v}^1(t), \dots,\alpha_{kj,v}^{L_v}(t)]$ being a sub-image allocation vector for sub-image $j$ in image $k$ of user $v$. 
%Considered the sub-image allocation, 
\nbl{Here, we assume that one sub-image cannot be split and transmitted over several sub-channels within the same time slot.}
The regenerated sub-image $\boldsymbol{s}_{k,v}^j(\boldsymbol{\alpha}_{kj,v}(t), \eta_{kj,v}(t))$ at user $v$ is 
\begin{align}\label{sub-image}
\boldsymbol{s}_{k,v}^j(\boldsymbol{\alpha}_{kj,v}(t), \eta_{kj,v}(t))=\left\{
\begin{aligned}
& \boldsymbol{s}_{k,v}^j,  \sum_{n=1}^N \boldsymbol{\alpha}_{kj,v}^n(t) =1,  \eta_{kj,v}(t) =1, \\
& 0,    \qquad\qquad\qquad\textrm{else} . 
% & 0, & \sum_{n=1}^N \boldsymbol{\alpha}_{kj}^n =0.
\end{aligned} 
\right. 
\end{align}
% \begin{equation}\label{sub-image}
% \boldsymbol{s}_k^j(\boldsymbol{\alpha}_{kj}(t), \eta_{kj}(t)) =
% \begin{cases}
% \boldsymbol{s}_k^j, & \sum_{n=1}^N \boldsymbol{\alpha}_{kj}^n(t) = 1,\ \eta_{kj}(t) = 1,\\
% 0,                  & \text{else}.
% \end{cases}
% \end{equation}
From~(\ref{sub-image}), we see that user $v$ can regenerate the original sub-image $\boldsymbol{s}_{k,v}^j$ when the BS uses a sub-channel to transmit the sub-image (i.e., $\sum_{n=1}^N\boldsymbol{\alpha}_{kj,v}^n(t) =1$) and user $v$ successfully receives it (i.e., $\eta_{kj,v} =1$).
In particular, $\alpha_{kj,v}^n(t) = 1$ implies that compressed semantic information $\hat{\boldsymbol{x}}_k^j$ is transmitted to user $v$ over sub-channel $n$ at time slot $t$, and $\alpha_{kj,v}^n(t) = 0$, otherwise.
% Thus, the recovered sub-image is
% \begin{equation}
%      = f_\lambda^{\textrm{D}}(\boldsymbol{x}^j_k),
% \end{equation}
% where $f_\lambda^{\textrm{D}}$ is the semantic decoder used to recover the sub-image with $\lambda$ being the decoder parameters.

\subsubsection{Sub-image Combination}
Each User $v$ combines the sub-images to regenerate the original image.
%Different from current works that directly summarize all sub-images~\cite{WU2024519},
Here, user $v$ will consider the maximum value of the same pixels in different sub-images as the value of the pixel in the regenerated image. This image regeneration method can preserve sharp boundaries since it can reduce the impact of zero elements added by the zero-padding method on sub-image regeneration.
%to reduce the impact of padding parts of the data
The value $\hat{i}_{k,p}$ of pixel $p$ of user $v$ in the regenerated image $\boldsymbol{S}_{k,v}$ is
\begin{equation}
    \hat{i}_{kp,v} = \textrm{max} \{ i_{kp, v}^1, i_{kp, v}^2, \dots, i_{kp, v}^J \}.
\end{equation}
%where $i_{k,p}^j$ is the value of the pixel $p$ in sub-image $\boldsymbol{s}_k^j$. 

\subsection{Feedback of Receiver}\label{sec:fb}
%CSI feedback and  preference feedback(score)
At each time slot, each user can 1) transmit either CSI feedback or content preference feedback to the transmitter, or 2) remain silent without sending any feedback. 
% One is the preference feedback, the other is the channel state information (CSI) feedback.
The content preference feedback transmitted by user $v$ at time slot $t$ is $\boldsymbol{D}_v(t) = [(\boldsymbol{u}_{1,v}(t), \omega_{1,v}(t)), (\boldsymbol{u}_{2,v}(t), \omega_{2,v}(t)), \dots, (\boldsymbol{u}_{D,v}(t), \omega_{D,v}(t))]$, where $\boldsymbol{u}_{d,v}$ is the content that user $v$ wants to receive and $\omega_{d,v}$ is the importance weight of content $\boldsymbol{u}_{d,v}$.
% each  key-value pair $(\boldsymbol{v}_d, \omega_d)$ is an interested category and its corresponding importance weight $\omega_d$.
Here, $\boldsymbol{D}_v(t)$ includes only the top $D$ categories that user $v$ wants to receive and user $v$ cannot send all content category information to the BS due to limited power.
% After receiving the object preference feedback, the BS allocates the sub-image by taking the importance weights into account (e.g., assigning more important sub-images to higher-quality channels).
% $\boldsymbol{p}(t) = [ p_1(t), p_2(t), \dots, p_J(t)]$ using the Sentence-T5 encoder~\cite{ni2021sentencet5scalablesentenceencoders}, where $p_j(t)$ denotes the preference score for sub-images $j$ of the original image. 
% Each score is obtained by calculating the semantic similarity between the category of the sub-image and the preference text given by the user using the Sentence-T5 encoder~\cite{ni2021sentencet5scalablesentenceencoders}. 
% If the score $p_j$ is higher, it indicates that the user pays more attention to the content of the sub-image. 
% If a sub-image has a high preference score, it implies that the receiver prefers to receive this sub-image more.
% Larger scores indicate greater interest in the associated sub-image.
% 内容偏好反馈包含了对原图J个类别的喜好评分。它是通过计算每个子类别与用户给出的偏好文字语义相似度得到。
% a text that indicates which category in the image the receiver is most interested in.
% Thus, the transmitter can prioritize specific objects and protect them accordingly by assigning them higher-quality sub-channels according to the preference feedback.
% The knowledge of accurate channel state information (CSI) at the transmitter is also essential to obtaining the transmission performance gains~\cite{9931713}.
The CSI feedback of user $v$ is defined as a vector $\boldsymbol{c}_v(t) = [c_{1,v}(t), c_{2,v}(t), \dots, c_{L_v,v}(t)]$ at time slot $t$. Here, $c_{n,v}(t) \in \{0,1\}$ indicates whether sub-image $\boldsymbol{s}^j_{k,v}$ is successfully transmitted to user $v$ over sub-channel $n$ in time slot $t$. 
%It can also reflect the channel gain $\boldsymbol{h}_n(t)$~\cite{perber}. 
Based on the CSI feedback, the BS can learn the channel conditions (e.g., the downlink channel gain $\boldsymbol{h}_{n,v}^{\textrm{D}}(t)$~\cite{perber}) and the sub-images that have been received by user $v$. 
Let $\beta_{\textrm{D},v}(t), \beta_{\textrm{C},v}(t) \in \{0,1\}$ be the feedback selection indicators where $\beta_{\textrm{D},v}(t)=1$ (e.g., $\beta_{\textrm{C},v}(t)=1$) implies that user $v$ will transmit content preference feedback (CSI feedback) to the BS at time slot $t$. Otherwise, we have $\beta_{\textrm{D},v}(t) = 0$ (e.g., $\beta_{\textrm{C},v}(t) = 0$). 
\subsection{Problem Formulation}
Given the defined system model, our objective is to minimize the sum of the users' semantic-weighted MSE between the original images and the images regenerated by the user.
% We aim to minimize the average MSE minization problem by jointly optimizing the sub-channel allocation $\widetilde{\boldsymbol{\alpha}} \triangleq \{ \boldsymbol{\alpha}_k \}_{k \in K}$, 
This minimization problem includes optimizing the sub-image allocation $\boldsymbol{\alpha} = \{\boldsymbol{\alpha}_1(t), \boldsymbol{\alpha}_2(t), \dots, \boldsymbol{\alpha}_K(t) \}_{t \in \mathcal{T}}$, user power allocation $\boldsymbol{p} = \{p^{\textrm{U}}_1(t), p^{\textrm{U}}_2(t), \dots, p_V^{\textrm{U}}(t) \}_{t \in \mathcal{T}}$ and user feedback selection $\boldsymbol{\beta} = \{\beta_{\textrm{D},1}(t),  \beta_{\textrm{C},1}(t), \dots, \beta_{\textrm{D},V}(t), \beta_{\textrm{C},V}(t) \}_{t \in \mathcal{T}}$.
We consider that the BS will transmit $K$ images to each user over a time period that consists of a set $\mathcal{T}$ of $K T$ time slots.
Each image must be transmitted within $T$ time slots.
The optimization problem is formulated as

	\begin{align}
    \min_{\boldsymbol{\alpha}, \boldsymbol{\beta}, \boldsymbol{p}} & 
    \sum_{v \in \mathcal{V}} \sum_{k \in \mathcal{K}} \sum_{j \in \mathcal{M}_k}\omega_{j,v}(t)\, \| \boldsymbol{m}^j_{k,v} - \boldsymbol{s}_{k,v}^j(\boldsymbol{\alpha_{kj,v}}(t),\eta_{kj,v}(t)) \|_2^2  \label{obj} \\ % 这里禁止编号
	\text{ s.t. } 
	&  \alpha_{kj,v}^n(t)\in \{0,1\}, \forall n \in \mathcal{L}_v, \forall j \in \mathcal{M}_{k,v}, \forall k \in \mathcal{K}, \forall v \in \mathcal{V}, \tag{\ref{obj}{a}} \label{alloc con} \\
    &  \alpha_{kj,v}^n(t) = 0, \forall t\in\mathcal{T}\setminus\mathcal{T}_k, \tag{\ref{obj}{b}}\label{slot-eq2} \\
    % & \sum_{t = t^{\prime}}^{t^{\prime}+qt}\sum_{j \in \mathcal{M}_k}\sum_{n \in \mathcal{N}} \alpha^n_{kj}(t) = J, \forall k \in \mathcal{K}, \\
    & \sum_{n \in \mathcal{L}_v}\sum_{t \in \mathcal{T}} \alpha^n_{kj,v}(t) \leq 1,  \forall j \in \mathcal{M}_k, \forall k \in \mathcal{K}, \forall v \in \mathcal{V}, \tag{\ref{obj}{c}} \label{block con}  \\
    & \sum_{j \in \mathcal{M}_{k,v}} \alpha^n_{kj,v}(t) \leq 1, \forall n \in \mathcal{L}_v, \forall t \in \mathcal{T}_k, \forall v \in \mathcal{V}, \tag{\ref{obj}{d}} \label{time con}  \\
    & \beta_{\textrm{D},v}(t) + \beta_{\textrm{C},v}(t) \le 1, \forall t \in \mathcal{T}, \forall v \in \mathcal{V}, \tag{\ref{obj}{e}} \label{fb con} \\
    % & \sum_{t \in \mathcal{T}}\beta_D(t) + \beta_c(t) \le B_{\textrm{max}}, \tag{\ref{obj}{f}}
    % \label{freq con} \\
    % & (\beta_{\textrm{D},v}(t) + \beta_{\textrm{C},v}(t)) P_v \tau_v(\boldsymbol{f}_v(t)) \leq E, \forall v \in \mathcal{V}, 
    % \tag{\ref{obj}{f}} \label{freq con}\\
    & \sum_{t \in \mathcal{T}} p_v^U(t) \leq p_{\text{max}}, \forall v \in \mathcal{V},    \tag{\ref{obj}{f}} \label{power con} \\
    & \frac{Z(\hat{\boldsymbol{x}}_{k,v}^j)}{R_{n,v}^{\textrm{D}}(t)} \le \bigtriangleup  t, \text{ if } \alpha_{kj, v}^n(t) = 1, \tag{\ref{obj}{g}} \label{rate con} \\
    & \beta_{\textrm{D},v}(t)Z( \boldsymbol{D}_v(t)) + \beta_{\textrm{C},v}(t)Z(\boldsymbol{c}_v(t)) \leq  R^{\textrm{U}}_v(t) \bigtriangleup t_f, \notag \\ 
    & \text{ if } \beta_{\textrm{D},v}(t) + \beta_{\textrm{C},v}(t) = 1 \tag{\ref{obj}{h}}, \label{fbt con} 
    % & p_v^U(t) \tau_v(\beta_{\textrm{D},v}(t),\beta_{\textrm{C},v}(t)) \leq E , \forall t \in \mathcal{T}, \forall v \in \mathcal{V} \tag{\ref{obj}{i}} \label{energy con}
    % & \sum_{t \in \mathcal{T}}\sum_{n \in \mathcal{N}} P_n(t) \leq P_{\textrm{max}}, \label{power con} \\
    % &\bar{\alpha}_{kj}^n \in\{0,1\},\quad \forall n\in\mathcal{N},\; j\in\mathcal{M}_k,\; k\in\mathcal{K}, \label{slotvar}\\
    % &\alpha_{kj}^n(t) = \bar{\alpha}_{kj}^n,\quad \forall t\in\mathcal{T}_k,\; \forall n,j,k, \label{slot-eq1}\\
	\end{align}
% \end{subequations}
% \nrd{SL: I noticed some inconsistencies in the notation, for example: $\beta_{\textrm{D},v}(t)$, $\beta_{\textrm{D},v}(t)$, $\beta_{\textrm{C},v}(t)$, $\beta_{\textrm{C},v}(t)$. Also, I defined some variables in~\eqref{eq:feedback-indicator}-\eqref{eq:feedback-success}. Feel free to integrate them into the optimization problem if you want.}
where $\mathcal{K}$ is the set of $K$ images that the BS must transmit to each user
%$\mathcal{M}_k $ is the set of $J_k$ sub-images in the original image $k$, 
and $\mathcal{T}_k \triangleq \{t \in \mathcal{T}: (k-1)T \leq t \leq kT \}$ is the set of time slots for transmitting the image $k$,
% $I^i_k$ and $S^i_k$ represent the pixel value $i$ of the preprocessed and recovery images, respectively. 
$\|\cdot\|_2$ denotes the $\ell_2$-norm, 
$Z(\hat{\boldsymbol{x}}_{k,v}^j)$ is the data size of semantic information $\hat{\boldsymbol{x}}_{k,v}^j$, 
$\bigtriangleup  t$ is the time duration of a time slot, 
$\bigtriangleup  t_f$ is the time duration of feedback transmission, 
% \begin{equation}
% $\tau_v(\beta_{\textrm{D},v}(t), \beta_{\textrm{C},v}(t)) = \frac{\beta_{\textrm{D},v}(t)Z( \boldsymbol{D}_v(t)) + \beta_{\textrm{C},v}(t)Z(\boldsymbol{c}_v(t))}{R^{\textrm{U}}_v(t) }$,
% \end{equation}
% where $Z(\boldsymbol{D}_v(t) )$ and $Z(\boldsymbol{c}_v(t))$ are the data sizes of content feedback and CSI feedback. These feedbacks will be introduce in Section~\ref{sec:fb}.
%$B_{\textrm{max}}$ is the maximum allowable number of feedback transmissions.
and $p_{\textrm{max}}$ is the available power of each user for feedback transmissions.
Constraints~(\ref{alloc con}), (\ref{slot-eq2}) and (\ref{block con}) ensure that the BS can transmit each sub-image once in specific time slots no matter whether the user $v$ correctly receives it.
Constraint~(\ref{time con}) implies that the BS can use only one time slot to transmit one sub-image.
Constraint~(\ref{fb con}) enforces that each user $v$ can only transmit one type of feedback to the BS per time slot.
Constraint~(\ref{power con}) limits the power used for feedback transmission with $p_{\text{max}}$ being users total power budget.
Constraints~(\ref{rate con}) and (\ref{fbt con}) are the requirements of each sub-image transmission and feedback transmission delay, respectively. 
% Constraint~(\ref{energy con}) limits the energy of transmitting feedback at each time slot with $E$ being the energy budget.
% Constraint~(\ref{power con}) limits the total transmit power to $P_{\textrm{max}}$.
% We also assume that  normalized time slot duration. Given that, the minimum data rate in each time slot $r_{q,n}^t$ is determined. Thus, we have
%(semantic encoder?)
% sub-image的生成与allocation，feedback选择存在隐式的关系，难以用准确的函数表述，同时约束较多，因此用传统优化算法求解几乎不可能，考虑到上述困难，我们提出用深度强化学习的方法进行求解

\nbl{The formulated problem in~(\ref{obj}) is challenging to solve via traditional optimization methods due to the following reasons.
First, from~(\ref{obj}), we see that the objective function depends not only on the sub-image allocation vector $\boldsymbol{\alpha}$ and power allocation vector $\boldsymbol{p}$, but also on variables such as user feedback selection $\boldsymbol{\beta}$, which cannot be expressed by explicit functions.
Second, since users can transmit a limited number of feedback messages, a tradeoff exists between content preference feedback and CSI feedback transmission. Here, more content preference feedback transmissions enable the BS to obtain more content categories that users are interested in, while more CSI feedback transmission enables the BS to learn more channel information (e.g., $\boldsymbol{h}_{n,v}^{\textrm{D}}(t)$) thus allocating appropriate content to each sub-channel. 
Third, due to limited energy, users cannot send feedback at each time slot and hence, the BS cannot know the real-time wireless environment and sub-image transmission status, which further complicates the optimization problem.
Finally, users must cooperatively transmit feedback since the bandwidth used for feedback transmission is limited. If all users transmit feedback simultaneously, they may not have enough bandwidth to ensure feedback be transmitted to the BS within the required delay.}
%with a limited number of feedback, the transmitter 
%Thus, the problem cannot be separated into smaller parts for traditional optimization methods.
%Fourth, the 
% we propose a masked reinforcement learning (RL) algorithm with dynamic neighborhood construction~\cite{akkerman2024dynamic} that enables the BS to schedule the sub-images based on their importance and the received feedback selected by the receiver so as to improve the MSE of all recovered images.
\section{Proposed Algorithm}\label{se:algorithm}
\begin{figure*}[t]
    \centering
    \includegraphics[width=15cm]{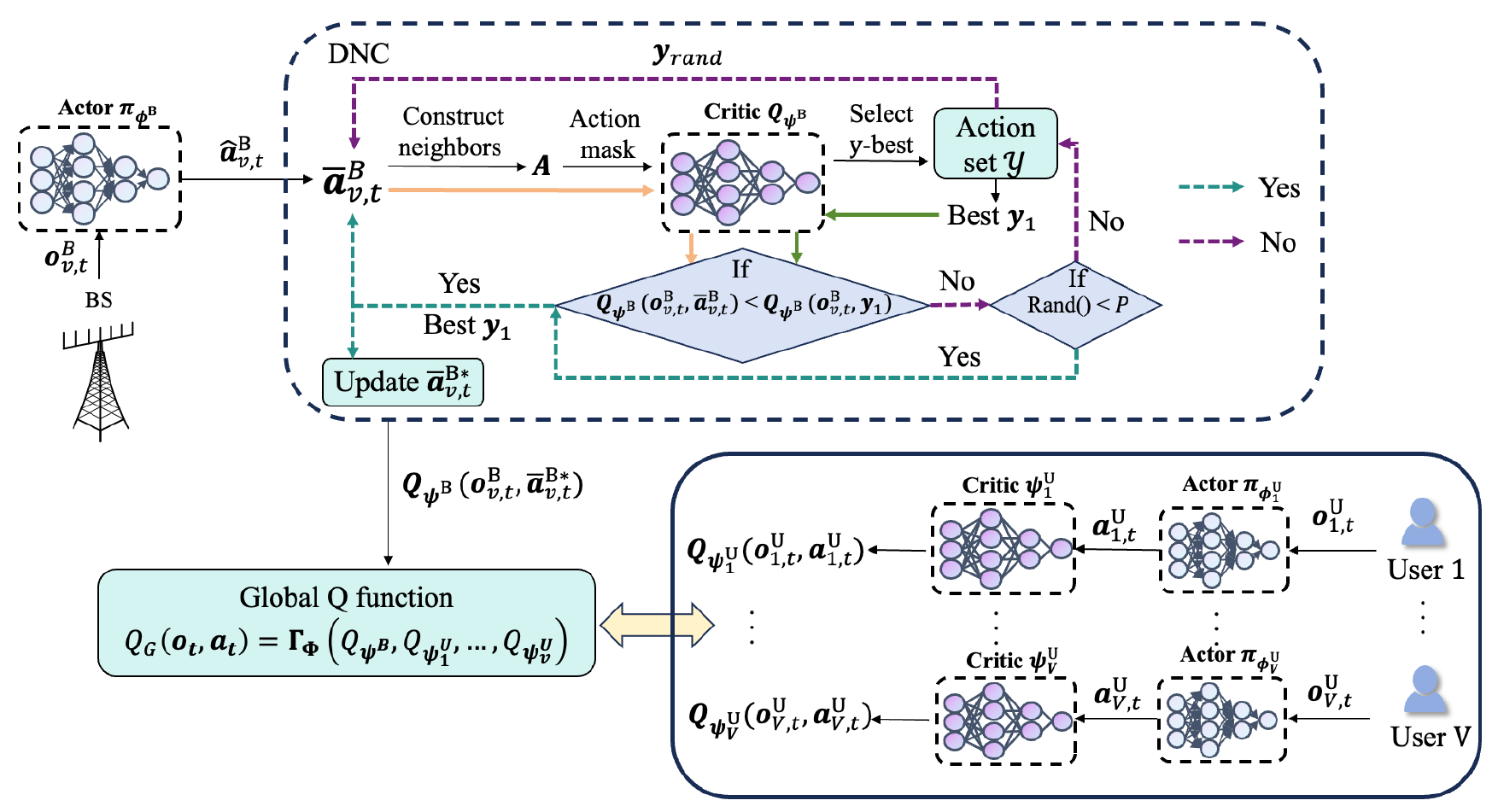}
    \nblcaption{The architecture of the proposed VDAC-DNC algorithm}
    \label{fig:algorithm}
    \vspace{-0.3cm}
\end{figure*}
% \begin{figure}[t]
%     \centering
%     \includegraphics[width=8cm]{Journal_Figure/algorithm.pdf}
%     \caption{The architecture of the proposed VDAC-DNC algorithm}
%     \label{fig:placeholder}
% \end{figure}
In this section, we introduce a VDAC-DNC algorithm to solve~(\ref{obj}).
% Compared to current MARL methods (e.g., multi-agent Q-learning like standard VDN~\cite{10.5555/3237383.3238080} has a higher algorithm performance with lower training efficiency, multi-agent actor-critic like MADDPG~\cite{lowe2017multi} has a higher training efficiency with lower algorithm performance), the proposed method combines the value decomposition paradigm and the actor critic paradigm, offering a reasonable trade-off between training efficiency and algorithm performance.
Compared with existing MARL methods such as VDN~\cite{10.5555/3237383.3238080} and MADDPG~\cite{10.5555/3295222.3295385}, the proposed approach combines the strengths of both value decomposition and actor–critic paradigms. Specifically, multi-agent Q-learning methods such as standard VDN achieve higher performance but with low training efficiency since they require extensive exploration and environment interactions, while multi-agent actor–critic methods such as MADDPG offer high training efficiency at the cost of reduced performance since they update the policy directly using gradient information and may easily get stuck in suboptimal solutions.
Thus, the proposed method offers a reasonable trade-off between training efficiency and algorithm performance.
Meanwhile, the proposed method is more effective in finding better actions in huge discrete action spaces since the designed method approximates the discrete action space using a continuous action and selects an action with the maximum $Q$ value from the continuous action space. 
%The actor generates a continuous action and uses critic-guided local search to pick a high-value discrete action with far less compute and memory.
% while enabling the agent to handle variable action dimensions by dynamically masking out invalid or irrelevant actions.
Next, we first introduce the components of the proposed VDAC-DNC method.
Then, we show the training process of the VDAC-DNC.
\subsection{Components of VDAC-DNC Method}
The components of the proposed VDAC-DNC method are specified as follows:

\textit{\textbf{Agent}}: 
%Our agent is the BS that determines sub-image transmission and user feedback selection.
The agents in the VDAC-DNC scheme are the BS and users.
At each time slot, the BS determines sub-image transmission, and then users determine whether to transmit a feedback and which type of feedback to transmit, as well as the power used for feedback transmission.
    
\textit{\textbf{BS States}}: 
The local state of the BS is used to describe the current observation of wireless environment and sub-images transmission status. The local state of the BS about user $v$ is defined as 
\begin{equation}
    \boldsymbol{o}_{v,t}^{\textrm{B}} = [\boldsymbol{z}_v(t), \boldsymbol{\omega}_v(t), \boldsymbol{c}_v(t), \boldsymbol{R}_v(t)],
\end{equation}
where $\boldsymbol{z}_v(t) = [Z(\hat{\boldsymbol{x}}_{k,v}^1), Z(\hat{\boldsymbol{x}}_{k,v}^2), \dots, Z(\hat{\boldsymbol{x}}_{k,v}^{J_{k,v}})]$ is a vector of the data sizes of the sub-images in an original image transmitted to user $v$,  $\boldsymbol{\omega}_v(t) = [\omega_{1,v}(t), \omega_{2,v}(t), \dots, \omega_{J_{k,v}}(t)]$ is the importance weight of the sub-images at the BS, $\boldsymbol{c}_v(t)$ is the CSI feedback, 
and $\boldsymbol{R}_v(t)= [R_{1,v}(t), R_{2,v}(t), \dots, R_{N,v}(t)]$ is the transmission rate of each channel.
Thus, the total local state of the BS is $\boldsymbol{o}_{\textrm{B}, t} =  [\boldsymbol{o}_{\textrm{B}, 1,t}, \boldsymbol{o}_{\textrm{B}, 2,t}, \dots,\boldsymbol{o}_{\textrm{B}, V,t}]$.
Note that the BS may not be able to receive the CSI and content preference feedback from the users at each time slot. Hence, when no feedback is received, $\boldsymbol{\omega}_v(t)$ and $\boldsymbol{c}_v(t)$ retain their values from the previous time slot (e.g., $\boldsymbol{\omega_v}(t) = \boldsymbol{\omega_v}(t-1)$, when $\beta_{\textrm{D},v}(t) = \beta_{\textrm{C}, v}(t) = 0$).

\textit{\textbf{User States}}: 
The local state of user $v$ is 
\begin{equation}
    \boldsymbol{o}_{v,t}^{\textrm{U}} =  [\hat{\boldsymbol{z}}_v(t), \hat{\boldsymbol{c}}_v(t), p^{\textrm{R}}_v(t)],
\end{equation}
where $\hat{\boldsymbol{z}}_v(t) = [Z(\boldsymbol{s}_{k,v}^1), Z(\boldsymbol{s}_{k,v}^2), \dots, Z(\boldsymbol{s}_{k,v}^{J_{k,v}})]$ is a vector of the data sizes of the received sub-images at user $v$, $\hat{\boldsymbol{c}}_v(t)$ is the current CSI, $p^{\textrm{R}}_v(t)$ is the remaining power of user $v$ at time slot $t$.
Hence, the joint state among all agents at time slot $t$ is a vector $\boldsymbol{o}_t = [\boldsymbol{o}_{t}^{\textrm{B}}, \boldsymbol{o}_{1,t}^{\textrm{U}}, \dots, \boldsymbol{o}_{v,t}^{\textrm{U}}]$.
    
\textit{\textbf{BS Actions}}: 
The actions of the BS at each step are sub-image allocation for each user. 
%and user feedback selection.
We define an action of the BS for user $v$ at time slot $t$ as $\boldsymbol{a}_{v,t}^{\textrm{B}} = [\delta_{1,v}(t), \delta_{2,v}(t), \dots, \delta_{N,v}(t)]$ with $\delta_{n,v}(t) \in \{0, J_{k,v}\}$ being the index of the sub-image transmitted to user $v$ and assigned to sub-channel $n$ at time slot $t$. 
% Note that the number of sub-image is different across images, thus the dimension of action is variable.
% We set the action dimension equal to the maximum number of sub-images, while masking those dimensions that exceed the actual number of sub-images in the current instance, as well as those actions that do not satisfy the required constraints~(\ref{slot-eq2}) --~(\ref{time con}) and ~(\ref{fb con}).

\textit{\textbf{User Actions}}: 
An action of each user includes the feedback selection and power allocation.
We define an action of the agent user $v$ at time slot $t$ as $\boldsymbol{a}_{\textrm{U}, v, t} = [\beta_{\textrm{D},v}(t), \beta_{\textrm{C},v}(t), p^{\textrm{U}}_v(t)]$.
Here, to guarantee that the selected actions meet the constraints~(\ref{slot-eq2}) --~(\ref{power con}), we use an action mask algorithm~\cite{Huang_2022} to remove all infeasible actions. 
The joint action vector of the BS and all users is $\boldsymbol{a}_t = [\boldsymbol{a}_{1,t}^{\textrm{B}}, \dots, \boldsymbol{a}_{V,t}^{\textrm{B}} \boldsymbol{a}_{1,t}^{\textrm{U}}, \dots, \boldsymbol{a}_{V,t}^{\textrm{U}}]$.

\textit{\textbf{Reward}}: 
% \nbl{Considering that directly using the negative weighted MSE as the reward will encourge the agents to transmit all the sub-images as soon as possible to minimize the MSE, which may violates the delay constraints~(\ref{rate con}) and~(\ref{fbt con}), we define the reward as the semantic-weighted MSE decrease rate.}
% The reward of choosing joint action $\boldsymbol{a}_t$ under a given state $\boldsymbol{o}_t$ is defined as the semantic-weighted MSE decrease rate. 
The proposed VDAC-DNC will maximize this reward, thus finding the optimal solution for the problem in~(\ref{obj}).
Therefore, the reward is
\begin{equation}
    r(\boldsymbol{o}_t,\boldsymbol{a}_t) = \sum_{v \in \mathcal{V}}D_{k,v}^{\textrm{B}}- D_{k,v}^{\textrm{A}}  - \kappa_{\textrm{R}} e_t^{\textrm{R}} - \kappa_{\textrm{F}} e_t^{\textrm{F}}, \label{team reward}
\end{equation}
where 
%the first term $D_{k,v}^{\textrm{B}}$ is the MSE of all sub-images before taking action $\boldsymbol{a}_t$, defined as
\begin{equation}
    D_{k,v}^{\textrm{B}}  =  \sum_{j \in \mathcal{M}_{k,v}} \omega_{j,v}(t) \| \boldsymbol{m}^j_{k,v} - \boldsymbol{s}_{k,v}^j(\boldsymbol{\alpha}_{kj,v}(t), \eta_{kj,v}(t)) \|_2^2, 
\end{equation}
is the semantic-weighted MSE of all sub-images before taking action $\boldsymbol{a}_t$.
The second term $D_{k,v}^{\textrm{A}}$ is the semantic-weighted MSE of all sub-images after taking action $\boldsymbol{a}_t$. The third term $e_t^{\textrm{R}} = \sum_{v \in \mathcal{V}} \sum_{n} \sum_{k,j} \mathbbm{1}{(\alpha^n_{kj,v}(t) \frac{Z(\hat{\boldsymbol{x}}_{k,v}^j)}{R_n(t)} > \bigtriangleup t)} $ is the number of sub-image transmissions whose delay exceeds the threshold. 
The final term $e_t^{\textrm{F}} = \sum_{v \in \mathcal{V}} \sum_{n} \sum_{k,j} \mathbbm{1}{((\beta_{\textrm{D},v}(t)+\beta_{\textrm{C},v}(t)) \tau_v > \bigtriangleup t_f)} $ is the number of feedback transmissions whose delay exceeds the threshold, $\kappa_{\textrm{R}}$ and $\kappa_{\textrm{F}}$ are weight parameters that control the weight penalty of sub-image transmission delay in~\eqref{rate con} and feedback transmission delay~\eqref{fbt con}, respectively. 
\nbl{Note that, we cannot directly use the negative weighted MSE as the reward since it will encourage the agents to transmit all the sub-images as soon as possible to minimize the MSE, which may violate the delay constraints~(\ref{rate con}) and~(\ref{fbt con})}
    
\textit{\textbf{BS Actor Networks}}: The actor network of BS $\boldsymbol{\pi}_{\boldsymbol{\phi}^{\textrm{B}}}(\hat{\boldsymbol{a}}_{v,t}^{\textrm{B}} | \boldsymbol{o}_{v,t}^{\textrm{B}})$ generates an action $\hat{\boldsymbol{a}}_{v,t}^{\textrm{B}}$ based on the current BS state $\boldsymbol{o}_{v,t}^{\textrm{B}}$.
Here, the action output by the actor network is different from the action defined in $\boldsymbol{a}_{v,t}^{\textrm{B}}$.
This is because the output of the actor network is a continuous action probability distribution and the action selected by the actor may not be in the discrete action space defined in $\mathcal{A}$.
For example, if we use one Gaussian distribution to approximate the continuous action probability distribution, the output of the actor network is a vector $\left[\mu, \sigma\right]$ where $\mu$ is the mean and $\sigma$ is the variance. Then, the action $\hat{\boldsymbol{a}}_{v,t}^{\textrm{B}}$ is sampled by this continuous Gaussian distribution (i.e., $\hat{\boldsymbol{a}}_{v,t}^{\textrm{B}} \sim \mathcal{N}(\mu, \operatorname{diag}(\sigma^2))$).
Hereinafter, we call the action selected by the actor as continuous action and the action defined in $\boldsymbol{a}_{v,t}^{\textrm{B}}$ as discrete action.
%定义continuous
%第一句话
%先介绍输出，
%定义清楚
%以自己的理解，写出来，不要引入过多的定义
% In particular, the actor network will first output a mean $\mu$ and a standard deviation $\sigma$ of a probability distribution. Then, the continuous action $\hat{\boldsymbol{a}}_t$ is sampled form this distribution. 
The continuous action $\hat{\boldsymbol{a}}_{v,t}^{\textrm{B}}$ will be used to find a discrete action $\boldsymbol{a}_{v,t}^{\textrm{B}}$ by using the DNC method, which will be introduced in Section~\ref{se:DNC}.
Note that $\hat{\boldsymbol{a}}_{v,t}^{\textrm{B}} \in \mathbb{R}^{N}$ is a continuous action that is used to construct the actual discrete action $\boldsymbol{a}_{v,t}^{\textrm{B}}$.

\textit{\textbf{User Actor Networks}}: The actor network of user $v$ $\boldsymbol{\pi}_{\boldsymbol{\phi}_{v}^{\textrm{U}}}(\boldsymbol{a}_{v,t}^{\textrm{U}} | \boldsymbol{o}_{v,t}^{\textrm{U}})$ generates an action $\boldsymbol{a}_{v,t}^{\textrm{U}}$ based on the current user $v$ state $\boldsymbol{o}_{v,t}^{\textrm{U}}$.
Since the DNC method can only process discrete action, as well as the action space of the user is small, we do not use the DNC method to process user actions. 

\textit{\textbf{The Critic Networks of the BS}}: 
The critic network of the BS is \st{a} three multi-layer perceptrons (MLPs) that approximate a local $Q$ function $Q_{\boldsymbol{\psi}^{\textrm{B}}}(\boldsymbol{o}_{v,t}^{\textrm{B}}, \boldsymbol{a}_{v,t}^{\textrm{B}})$ to find a discrete action $\boldsymbol{a}_{v,t}^{\textrm{B}}$ from the continuous action $\hat{\boldsymbol{a}}_{v,t}^{\textrm{B}}$ and then estimate the expected cumulative reward from the given state $\boldsymbol{o}_{v,t}^{\textrm{B}}$ and action $\boldsymbol{a}_{v,t}^{\textrm{B}}$, where $\boldsymbol{\psi}^{\textrm{B}}$ are the parameter vectors of the MLPs.

% and historical state-action pair $\boldsymbol{\chi}_{\textrm{B},t-1}$ at the previous time slot. 
\textit{\textbf{The Critic Networks of users}}:
The critic network of a user is also MLPs that approximates a local $Q$ function $Q_{\boldsymbol{\psi}_{v}^{\textrm{U}}}(\boldsymbol{o}_{v,t}^{\textrm{U}}, \boldsymbol{a}_{v,t}^{\textrm{U}})$ to estimate the expected cumulative reward from the given state $\boldsymbol{o}_{v,t}^{\textrm{U}}$ and action $\boldsymbol{a}_{v,t}^{\textrm{U}}$, where $\boldsymbol{\psi}_{v}^{\textrm{U}}$ are the parameter vectors of the MLPs.
%In particular, the local critic network for the BS and each user is approximated by with parameter vectors $\boldsymbol{\psi}^{\textrm{B}}$ and .

\textit{\textbf{Global Critic Networks}}: 
The global critic network $Q$ function $Q_{\textrm{G}}(\boldsymbol{o}_t, \boldsymbol{a}_t)$ is used to estimate the total rewards under a joint state $\boldsymbol{o}_t$ and a joint action $\boldsymbol{a}_t$ of all agents.
Here we use the critic network of the BS and users to approximate a global $Q$ function in two different methods.
First, we can summarize the $Q$ functions of all agents to approximate the global $Q$ function as done in VDN method~\cite{10.5555/3237383.3238080}:
\begin{equation}
    Q_{\textrm{G}}(\boldsymbol{o}_t, \boldsymbol{a}_t) = \sum_v Q_{\boldsymbol{\psi}^{\textrm{B}}}(\boldsymbol{o}_{v,t}^{\textrm{B}}, \boldsymbol{a}_{v,t}^{\textrm{B}}) + Q_{\boldsymbol{\psi}_{v}^{\textrm{U}}}(\boldsymbol{o}_{v,t}^{\textrm{U}}, \boldsymbol{a}_{v,t}^{\textrm{U}}). \label{sumgq}
\end{equation}
Second, we can use a MLP model with parameters $\boldsymbol{\Gamma}_{\boldsymbol{\Phi}}$ to aggregate local $Q$ functions of all agents as done in QMIX~\cite{JMLR:v21:20-081}. 
% The mixing network is realized with feedforward neural networks (FNNs) with parameters $\boldsymbol{\Phi}$, named hypernetwork. 
Hence, the global $Q$ function $Q_{\textrm{G}}(\boldsymbol{o}_t, \boldsymbol{a}_t)$ is 
\begin{equation}
    Q_{\textrm{G}}(\boldsymbol{o}_t, \boldsymbol{a}_t) = \boldsymbol{\Gamma}_{\boldsymbol{\Phi}}(Q_{\boldsymbol{\psi}^{\textrm{B}}}, Q_{\boldsymbol{\boldsymbol{\psi}_{1}^{\textrm{U}}}}, \dots, Q_{\boldsymbol{\boldsymbol{\psi}_{v}^{\textrm{U}}}}).  \label{qmix}
\end{equation}

\subsection{Training of the VDAC-DNC} \label{se:DNC}
Next, we introduce the entire training process of the proposed VDAC-DNC algorithm to solve problem~(\ref{obj}).
% Since the DNC method can only process discrete action, as well as the action space of the user is small, thus we do not apply the DNC method to process the user action. 
We first introduce the use of DNC method to select the BS discrete action $\boldsymbol{a}_{v,t}^{\textrm{B}}$ that can achieve the highest $Q$ value among the neighbors of the continuous action $\hat{\boldsymbol{a}}_{v,t}^{\textrm{B}}$ generated by the BS actor network.
Since the DNC method can only process discrete action, as well as the action space of the user is small, thus we do not apply the DNC method to process the user action. 
Then, we explain the loss function of the VDAC-DNC scheme and the training process.

%Since the action space is large, we use DNC, a continuous-to-discrete action mapping which requires no a priori definition of the action space.
% Since the discrete action space is large, to avoid full action enumeration, we can first learn a continuous policy and utilize DNC, a continuous-to-discrete action mapping~\cite{chandak2019learningactionrepresentationsreinforcement} to find a local higher-value action.
\begin{algorithm}[t]
	%\textsl{}\setstretch{1.8}
	\caption{Dynamic Neighborhood Construction}
    \small
	\renewcommand{\algorithmicrequire}{\textbf{Initialize:}}
	\renewcommand{\algorithmicensure}{\textbf{Input:}}
	\label{DNC}
	\begin{algorithmic}[1]
		% \ENSURE The channel state information $ \mathbf{h}_{A}^{k}(t),\mathbf{G}_U(t),\mathbf{h}_{I}^{k}(t),$ $\mathbf{h}_{E}^{k}(t),\mathbf{G}_E(t) $
		\REQUIRE $y$, $\epsilon$, $c_{\beta}$, $\bar{\boldsymbol{a}}_{v,t}^{\textrm{B}} \leftarrow \textrm{round}(\hat{\boldsymbol{a}}_{v,t}^{\textrm{B}}), \bar{\boldsymbol{a}}_{v,t}^{\textrm{B}\star} \leftarrow \bar{\boldsymbol{a}}_{t}^{\textrm{B}}, \mathcal{Y} = \emptyset$
		%\STATE  Calculate the order $\zeta_{i}$ according to $\mathbf{V}_{\text{init}}$
		\WHILE	{$y > 0$, or $\epsilon >0$}
        \STATE  Find neighbors $\boldsymbol{A}$ to $\bar{\boldsymbol{a}}_{v,t}^{\textrm{B}}$ with $P_{ij}$
        \STATE  Get Q-values for all neighbors in $\boldsymbol{A}$
        \STATE  Obtain set $\mathcal{Y}^{\prime}$ with $y$-best neighbors, $\mathcal{Y} \leftarrow \mathcal{Y} \cup \mathcal{Y}^{\prime}$
        \STATE  $\boldsymbol{y}_1 \leftarrow \mathcal{Y}^{\prime}$, with $\boldsymbol{y}_1 = \arg \max_{\boldsymbol{y} \in \mathcal{Y}^{\prime}} Q_{\boldsymbol{\psi}^{\textrm{B}}}(\boldsymbol{o}_{v,t}^{\textrm{B}}, \boldsymbol{y})$
        \IF {$Q_{\boldsymbol{\psi}^{\textrm{B}}}(\boldsymbol{o}_{v,t}^{\textrm{B}}, \boldsymbol{y}_1) > Q_{\boldsymbol{\psi}^{\textrm{B}}}(\boldsymbol{o}_{v,t}^{\textrm{B}}, \bar{\boldsymbol{a}}_{\textrm{B}, t}))$}
        \STATE Accept $\boldsymbol{y}_1 \in \mathcal{Y}^{\prime}, \bar{\boldsymbol{a}}_{v,t}^{\textrm{B}} \leftarrow \boldsymbol{y}_1$
        \IF{$Q_{\boldsymbol{\psi}^{\textrm{B}}}(\boldsymbol{o}_{v,t}^{\textrm{B}}, \boldsymbol{y}_1) > Q_{\boldsymbol{\psi}^{\textrm{B}}}(\boldsymbol{o}_{v,t}^{\textrm{B}}, \bar{\boldsymbol{a}}_{v,t}^{\textrm{B}\star})$ }
        \STATE $\bar{\boldsymbol{a}}_{v,t}^{\textrm{B}\star} \leftarrow \boldsymbol{y}_1$
        \ENDIF
        \ELSIF {rand() < $\exp [-(Q_{\boldsymbol{\psi}^{\textrm{B}}}(\boldsymbol{o}_{v,t}^{\textrm{B}}, \boldsymbol{y}_1)-Q_{\boldsymbol{\psi}^{\textrm{B}}}(\boldsymbol{o}_{v,t}^{\textrm{B}}, \boldsymbol{y}_1))/ \epsilon )]$}
        \STATE Accept $\boldsymbol{y}_1 \in \mathcal{Y}^{\prime}, \bar{\boldsymbol{a}}_{v,t}^{\textrm{B}} \leftarrow \boldsymbol{y}_1, \epsilon \leftarrow \epsilon - c_{\beta} $
        \ELSE 
        \STATE Reject $\boldsymbol{y}_1 \in \mathcal{Y}^{\prime}, \bar{\boldsymbol{a}}_{v,t}^{\textrm{B}} \leftarrow \boldsymbol{y}_{\textrm{rand}} \in \mathcal{Y} $
        \ENDIF
		\STATE $y \leftarrow \left \lceil y - c_k \right \rceil $
		%\STATE $\zeta_{\eta+1} \leftarrow \mathbf{V}_{\text{opt}}$
		\ENDWHILE
        \STATE Return $\bar{\boldsymbol{a}}_{v,t}^{\textrm{B}\star}$
	\end{algorithmic}
\end{algorithm}

\begin{algorithm}[t]
	%\textsl{}\setstretch{1.8}
	\caption{VDAC-DNC for semantic-weighted MSE minimization}
    \small
	\renewcommand{\algorithmicrequire}{\textbf{Initialize:}}
	\renewcommand{\algorithmicensure}{\textbf{Input:}}
	\label{VDAC}
	\begin{algorithmic}[1]
		% \ENSURE The channel state information $ \mathbf{h}_{A}^{k}(t),\mathbf{G}_U(t),\mathbf{h}_{I}^{k}(t),$ $\mathbf{h}_{E}^{k}(t),\mathbf{G}_E(t) $
		\REQUIRE Local $Q$ function $Q_{\boldsymbol{\psi}^{\textrm{B}}}$ of BS and $Q_{\boldsymbol{\psi}_{v}^{\textrm{U}}}$ of each user $v$, parameters $\boldsymbol{\Phi}$ of $\boldsymbol{\Gamma}_{\boldsymbol{\Phi}}$.
    \FOR{each episode}
        \FOR{each user $v$}
        \STATE BS observes the local state $\boldsymbol{o}_{v,t}^{\textrm{B}}$
        \STATE Generate a continuous action $\hat{\boldsymbol{a}}_{v,t}^{\textrm{B}} \leftarrow \boldsymbol{\pi}_{\boldsymbol{\phi}^{\textrm{B}}}(\boldsymbol{o}_{v,t}^{\textrm{B}})$ using the BS actor network.
        \STATE $\bar{\boldsymbol{a}}_{v,t}^{\textrm{B}} \leftarrow \hat{\boldsymbol{a}}_{v,t}^{\textrm{B}}, \bar{\boldsymbol{a}}_{v,t}^{\textrm{B}\star} \leftarrow \textrm{DNC}(\bar{\boldsymbol{a}}_{v,t}^{\textrm{B}})$
        \STATE Calculate local $Q_{\boldsymbol{\psi}^{\textrm{B}}}(\boldsymbol{o}_{v,t}^{\textrm{B}}, \boldsymbol{a}_{v,t}^{\textrm{B}})$
        \ENDFOR
        \FOR{each user $v = 1,2, \dots, V$}
        \STATE User $v$ observes the local state $\boldsymbol{o}_{v,t}^{\textrm{U}}$
        \STATE Select an action $\boldsymbol{a}_{\textrm{U}, v,t}$ using the user actor network.
        \STATE Calculate local $Q_{\boldsymbol{\psi}_{v}^{\textrm{U}}}(\boldsymbol{o}_{v,t}^{\textrm{U}}, \boldsymbol{a}_{v,t}^{\textrm{U}})$
        \ENDFOR
        \STATE Calculate the loss $L_c(\boldsymbol{\psi}_{v}^{\textrm{U}},\boldsymbol{\psi}^{\textrm{B}}, \boldsymbol{\Phi})$ 
        \STATE The BS updates its actor network using \eqref{eq:bs-actor-update-short} and updates its critic network using \eqref{bscri}
        \FOR{each user $v$}
        \STATE Update the local actor network using \eqref{eq:user-actor-update-short} and update the critic network using \eqref{usercri}
        \ENDFOR
        \STATE Update the hypernetwork using \eqref{hyper}
        \ENDFOR
	\end{algorithmic}
\end{algorithm}
\subsubsection{Discrete action selection} 
% To update the critic and actor networks, the actor generates a continuous action $\hat{\boldsymbol{a}}_t$ based on the current state and the actor $\boldsymbol{\pi}_{\boldsymbol{\phi}}$. Then
The procedure of using the DNC method to find a discrete action is summarized as follows:

\textit{\textbf{Step 1}}: The BS uses the continuous action $\hat{\boldsymbol{a}}_{v,t}^{\textrm{B}}$ to generate a discrete action $\bar{\boldsymbol{a}}_{v,t}^{\textrm{B}}$ by rounding each element in $\hat{\boldsymbol{a}}_{v,t}^{\textrm{B}}$ to its closest discrete value, which is expressed as 
\begin{equation}
\bar{\boldsymbol{a}}_{v,t}^{\textrm{B}} = \left \lfloor \hat{\boldsymbol{a}}_{v,t}^{\textrm{B}} + \frac{1}{2}  \right \rfloor ,
\end{equation}
where $\left \lfloor  \right \rfloor$ is the floor function. 

\textit{\textbf{Step 2}}: The BS searches the neighboring actions of $\bar{\boldsymbol{a}}_{v,t}^{\textrm{B}}$ to find a better discrete action. In particular, we first define a perturbation matrix $\boldsymbol{P} = (P_{ij})_{i=1,\ldots,U;\, j=1,\ldots,2lU} \in \mathbb{R}^{U \times 2 l U}$, where $U$ is the action dimension and $l$ is the range of neighboring actions we will search for each action.
Its first $lU$ columns represent the positive search of the discrete action $\bar{\boldsymbol{a}}_{v,t}^{\textrm{B}}$. In particular, for each column, one element of the discrete action will be increased by step size $\epsilon$.
The last $lU$ columns of the matrix represent negative movements of the discrete action $\bar{\boldsymbol{a}}_{v,t}^{\textrm{B}}$.
In particular, if $ l =1$,
the matrix $\boldsymbol{P}$ is
\begin{equation}
   \boldsymbol{P}=\left[\begin{array}{cccccccc}
\epsilon & 0 & \cdots & 0 & -\epsilon & 0 & \cdots & 0 \\
0 & \epsilon & \cdots & 0 & 0 & -\epsilon & \cdots & 0 \\
\vdots & \vdots & \ddots & \vdots & \vdots & \vdots & \ddots & \vdots \\
0 & 0 & \cdots & \epsilon & 0 & 0 & \cdots & -\epsilon
\end{array}\right] .
\end{equation}
 Let $\bar{\boldsymbol{A}} = [\bar{\boldsymbol{a}}_{v,t}^{\textrm{B}},\bar{\boldsymbol{a}}_{v,t}^{\textrm{B}}, \dots,\bar{\boldsymbol{a}}_{v,t}^{\textrm{B}}] \in \mathbb{Z}^{U \times 2lU}$ be a matrix that consists of $2lU$ action $\bar{\boldsymbol{a}}_{v,t}^{\textrm{B}}$. Then, $2lU$ neighboring actions of $\bar{\boldsymbol{a}}_{v,t}^{\textrm{B}}$ are represented by
\begin{equation}
\boldsymbol{A} = \bar{\boldsymbol{A}} + \boldsymbol{P} = [\bar{\boldsymbol{a}}_{v,t,1}^{\textrm{B}}, \bar{\boldsymbol{a}}_{v,t,2}^{\textrm{B}}, \dots, \bar{\boldsymbol{a}}_{v,t,2lU}^{\textrm{B}}]. \label{newmat}
\end{equation}
From~\eqref{newmat}, we see that each neighboring action in $\boldsymbol{A}$ only has one element that is different from $\bar{\boldsymbol{a}}_{v,t}^{\textrm{B}}$, which guarantees that the action search will only move in one direction per step.

\textit{\textbf{Step 3}}: We utilize the critic to obtain $Q$-values for neighboring actions in $\boldsymbol{A}$. 
Let $\mathcal{Y}$ be the initial overall action set.
Among these neighbors, we select the $y$ actions with the largest $Q$-values and sort them in descending order, forming an ordered action subset $\mathcal{Y}^{\prime}$.
We then add $\mathcal{Y}^{\prime}$ to $\mathcal{Y}$.
% Then we store the top $y$ neighboring actions in an ordered set $\mathcal{Y}^{\prime}$ and store $\mathcal{Y}^{\prime}$ in $\mathcal{Y}$.
From $\mathcal{Y}^{\prime}$, we select action $\boldsymbol{y}_1$, which has the highest $Q$-value in $\mathcal{Y}^{\prime}$.

\textit{\textbf{Step 4}}: We use an iterative simulated annealing (SA)-based~\cite{KochenderferWheeler2019} search scheme to explore the neighborhoods in $\boldsymbol{A}$ and find the action $\boldsymbol{a}_{v,t}^{\textrm{B}\star}$ that can achieve the highest $Q$ value. SA is a probabilistic search method which performs worse actions instead of performing the best actions in a specific probability.
In particular, throughout the search process, we reduce $\epsilon$ by adjusting the cooling parameter $c_{\beta}$, thus escaping local optima and becoming more robust than greedy search.
In particular, if the $Q$-value of $\boldsymbol{y}_1 \in \mathcal{Y}^{\prime}$ exceeds the $Q$-value of the base action $\bar{\boldsymbol{a}}_{v,t}^{\textrm{B}}$, we accept $\boldsymbol{y}_1$ as new discrete action $\bar{\boldsymbol{a}}_{v,t}^{\textrm{B}}$. If not, we still accept it with probability $\textrm{exp}[-(Q_{\boldsymbol{\psi}_{B}}(\boldsymbol{o}_{v,t}^{\textrm{B}},\bar{\boldsymbol{a}}_{v,t}^{\textrm{B}})- Q_{\boldsymbol{\psi}_{B}}(\boldsymbol{o}_{v,t}^{\textrm{B}},\boldsymbol{y}_{1})) / \epsilon]$.

Steps 2-4 are repeated until the algorithm reaches the predefined number of iterations.
Finally, we obtain the action $\bar{\boldsymbol{a}}^{\textrm{B}\star}_{v,t}$ with the highest $Q$-value.
% The process is repeated until the parameter of temperature is less than 0, then we obtain an action  and set $\boldsymbol{a}_t = \bar{\boldsymbol{a}}^{\star}_t$.
Then the BS combines actions $\bar{\boldsymbol{a}}^{\textrm{B}\star}_{v,t}$ and uses action $\bar{\boldsymbol{a}}^{\star}_{\textrm{B},t}$ to interact with the user.
% and collects a transition $\{(\boldsymbol{o}_t, \boldsymbol{a}_t, r_t(\boldsymbol{s}_t, \boldsymbol{a}_t), \boldsymbol{s}_{t+1})\}$ and stores it in a replay buffer for RL model training.
The specific process of DNC algorithm is summarized in Algorithm~\ref{DNC}.
\subsubsection{Critic network and actor network update}
Next, we introduce the update of the local critic, the actor networks, and the MLP used to aggregate local $Q$ functions.
% The objective training AC-DNC is to maximize the expected cumulative reward, defined as 
% \begin{equation}
%     \chi(\phi) = \mathbb{E}_{\boldsymbol{\pi}_{\phi}}[\sum_{t=1}^T]
% \end{equation}
\begin{itemize}
    \item \textit{Local Critic Network Update}: \nbl{The local critic network at a user or the BS is updated by minimizing the temporal difference (TD) error, which is defined as}
    \begin{align}
        &L_c(\boldsymbol{\psi}_{v}^{\textrm{U}},\boldsymbol{\psi}^{\textrm{B}}, \boldsymbol{\Phi})  = \notag \\ &\mathbb{E}[r_t(\boldsymbol{o}_t, \boldsymbol{a}_t) + \gamma Q_{\textrm{G}}(\boldsymbol{o}_{t+1}, \boldsymbol{a}_{t+1} )   - Q_{\textrm{G}}(\boldsymbol{o}_t, \boldsymbol{a}_t)]. \label{critic}
    \end{align}
    Given~(\ref{critic}), the parameters of local critic network is updated by a stochastic gradient descent (SGD) method as follows:
    % \begin{equation}
    % \end{equation}
    \begin{align}
        \boldsymbol{\psi}^{\textrm{B}} & \xleftarrow{} \boldsymbol{\psi}^{\textrm{B}} - \lambda_{c} \nabla_{\boldsymbol{\psi}^{\textrm{B}}} \frac{1}{2}L_c^2(\boldsymbol{\psi}_{v}^{\textrm{U}},\boldsymbol{\psi}^{\textrm{B}}, \boldsymbol{\Phi}), \label{bscri} \\
        \boldsymbol{\psi}_{v}^{\textrm{U}} & \xleftarrow{} \boldsymbol{\psi}_{v}^{\textrm{U}} - \lambda_{c} \nabla_{\boldsymbol{\psi}_{v}^{\textrm{U}}} \frac{1}{2}L_c^2(\boldsymbol{\psi}_{v}^{\textrm{U}},\boldsymbol{\psi}^{\textrm{B}}, \boldsymbol{\Phi}), \label{usercri}
    \end{align}
    where $\lambda_{c}$ is the learning rate, $\nabla_{\boldsymbol{\psi}^{\textrm{B}}} L_c(\boldsymbol{\psi}^{\textrm{B}})$ and $\nabla_{\boldsymbol{\psi}_{v}^{\textrm{U}}} L_c(\boldsymbol{\psi}_{v}^{\textrm{U}})$ are the gradients of the critic networks of the BS and users.
    \item \textit{Actor Update}: The actor network parameters of both BS and the users are updated using policy gradient ascent, as follows:
    % \begin{equation}
    %     \boldsymbol{\phi}^{\textrm{B}} \xleftarrow{} \boldsymbol{\phi}^{\textrm{B}} + \lambda_{a} L_c \sum_{v \in \mathcal{V}} \nabla_{\boldsymbol{\phi}^{\textrm{B}}} \text{log}  \boldsymbol{\pi}_{\boldsymbol{\phi}^{\textrm{B}}}(\hat{\boldsymbol{a}}_{v,t}^{\textrm{B}} | \boldsymbol{o}_{v,t}^{\textrm{B}}), \label{bsact}
    % \end{equation}
\begin{align}
  \boldsymbol{\phi}^{\textrm{B}}
    &\leftarrow
      \boldsymbol{\phi}^{\textrm{B}}
      + \lambda_a \,L_c
        \sum_{v\in\mathcal V}
        \nabla_{\boldsymbol{\phi}^{\textrm{B}}}
        \log
        \pi_{\boldsymbol{\phi}^{\textrm{B}}}^{\textrm{B}}
          (\hat{\boldsymbol a}^{\textrm{B}}_{v,t}
           |\boldsymbol o^{\textrm{B}}_{v,t}),
  \label{eq:bs-actor-update-short} \\
  \boldsymbol{\phi}_v^{\textrm{U}}
    &\leftarrow
      \boldsymbol{\phi}_v^{\textrm{U}}
      + \lambda_a \,L_c\,
        \nabla_{\boldsymbol{\phi}_v^{\textrm{U}}}
        \log
        \pi_{\boldsymbol{\phi}_v^{\textrm{U}}}^{\textrm{U}}
          (\boldsymbol a^{\textrm{U}}_{v,t}
           |\boldsymbol o^{\textrm{U}}_{v,t}),
  \label{eq:user-actor-update-short}
\end{align}
    where $\lambda_{a}$ is the learning rate, and $\nabla_{\boldsymbol{\phi}^{\textrm{B}}} \text{log} \boldsymbol{\pi}_{\boldsymbol{\phi}^{\textrm{B}}}(\hat{\boldsymbol{a}}_t |\boldsymbol{o}_t)$ and $\nabla_{\boldsymbol{\phi}_v^{\textrm{U}}}
        \log
        \pi_{\boldsymbol{\phi}_v^{\textrm{U}}}^{\textrm{U}}$ are the gradients of the actor networks of the BS and users.
    % Note that using off-policy information in the actor weight update does not impact on learning stability.
    %The user actor network updates its policy using policy gradient ascent
    % \begin{equation}
    %     \boldsymbol{\phi}_{v}^{\textrm{U}} \xleftarrow{} \boldsymbol{\phi}_{v}^{\textrm{U}} + \lambda_{a}L_c \sum_{v \in \mathcal{V}}  \nabla_{\boldsymbol{\phi}_{v}^{\textrm{U}}} \text{log} \boldsymbol{\pi}_{\boldsymbol{\phi}_{v}^{\textrm{U}}}(\boldsymbol{a}_{v,t}^{\textrm{U}} | \boldsymbol{o}_{v,t}^{\textrm{U}}), \label{useract}
    % \end{equation}
    \item \textit{MLP Update}: The parameters of the MLP used to approximate the global $Q$ function is also updated using gradient ascent, as follows:
    \begin{equation}
        \boldsymbol{\Phi} \leftarrow \boldsymbol{\Phi} - \lambda_h \nabla_{\boldsymbol{\Phi}}L_c(\boldsymbol{\psi}_{v}^{\textrm{U}},\boldsymbol{\psi}^{\textrm{B}}, \boldsymbol{\Phi}), \label{hyper}
    \end{equation}
    where $\lambda_h$ is the learning rate.
\end{itemize}

%The specific procedure of the 

\section{Convergence, Implementation, and Complexity Analysis}\label{se:analy}
Next, we analyze the convergence, implementation, communication overhead, and complexity of the proposed VDAC-DNC method.
\subsection{Convergence Analysis}

%We now outline the convergence of the proposed VDAC-DNC algorithm.
To analyze the convergence of the proposed VDAC-DNC algorithm, we first define $\boldsymbol{\theta} =\big(\boldsymbol{\phi}^{\textrm{B}},\{\boldsymbol{\phi}_v^{\textrm{U}}\}_{v\in\mathcal V}\big)$
as the parameters of the BS and users.
%, and
%$\eta =\big(\boldsymbol{\psi}^{\textrm{B}},\{\boldsymbol{\psi}_v^{\textrm{U}}\}_{v\in\mathcal V},\boldsymbol{\Phi}\big)$ be the vector of the parameters of all critics and the MLP.
The joint policy factorizes as
\begin{equation}
  \pi_{\boldsymbol{\theta}}(\boldsymbol a_t|\boldsymbol o_t)
  = \prod_{v\in\mathcal V}
      \pi_{\boldsymbol{\phi}^{\textrm{B}}}^{\textrm{B}}
        (\boldsymbol a^{\textrm{B}}_{v,t}|\boldsymbol o^{\textrm{B}}_{v,t})
      \pi_{\boldsymbol{\phi}_v^{\textrm{U}}}^{\textrm{U}}
        (\boldsymbol a^{\textrm{U}}_{v,t}|\boldsymbol o^{\textrm{U}}_{v,t}),
  \label{eq:joint-policy}
\end{equation}
% where the BS actions $\boldsymbol a^{\textrm{B}}_{v,t}$ are obtained
% from continuous proto-actions
% $\hat{\boldsymbol a}^{\textrm{B}}_{v,t}
%  \sim \pi_{\boldsymbol{\phi}^{\textrm{B}}}^{\textrm{B}}
%        (\cdot|\boldsymbol o^{\textrm{B}}_{v,t})$
% via the deterministic DNC mapping.
% Since this mapping does not depend on~$\boldsymbol{\theta}$, it
% can be considered into the environmental dynamics and does not change the
% form of the policy gradient.
The global action-value function under policy $\pi_{\boldsymbol{\theta}}$ is denoted by
$Q^{\pi_{\boldsymbol{\theta}}}(\boldsymbol o,\boldsymbol a)$ and is approximated by the
value-decomposition critic $Q_\textrm{G}(\boldsymbol{o}_t, \boldsymbol{a}_t)$.
%\paragraph*{Global critic and TD error}
% \begin{equation}
%   Q_{\mathrm G}(\boldsymbol o_t,\boldsymbol a_t;\eta)
%   = \boldsymbol\Gamma_{\boldsymbol\Phi}\big(
%       Q_{\boldsymbol{\psi}^{\textrm{B}}},
%       Q_{\boldsymbol{\psi}_1^{\textrm{U}}},
%       \ldots,
%       Q_{\boldsymbol{\psi}_V^{\textrm{U}}}
%     \big),
% \end{equation}
% The one-step temporal-difference (TD) error is
% \begin{equation}
%   \delta_t
%   = r_t(\boldsymbol o_t,\boldsymbol a_t)
%     + \gamma Q_{\mathrm G}(\boldsymbol o_{t+1},\boldsymbol a_{t+1};\eta)
%     -       Q_{\mathrm G}(\boldsymbol o_t,\boldsymbol a_t;\eta),
%   \label{eq:td-error}
% \end{equation}
% and the critic parameters $\eta$ are trained by minimizing the squared
% TD loss
% \begin{equation}
%   L_c(\eta) = \mathbb{E}[\delta_t^2],
% \end{equation}
% leading to the stochastic gradient-descent updates in
% \eqref{bscri}--\eqref{hyper}.
%\paragraph*{Policy-gradient expression}
We consider the discounted return objective
\begin{equation}
    J(\boldsymbol{\theta})
  = \mathbb{E}_{\pi_{\boldsymbol{\theta}}}\Big[\sum_{t=0}^{\infty} \gamma^t r(\boldsymbol{o}_t,\boldsymbol{a}_t)\Big], \label{return}
\end{equation}
where  $\gamma$ is the discount factor that weights future rewards. Then, the convergence of the designed method is analyzed in the following lemma.
% Let $V_{\mathrm{tot}}(\boldsymbol o)$ be a state-value baseline (e.g.,
% $V_{\text {tot}}(\boldsymbol{o}) \triangleq \mathbb{E}_{\boldsymbol{a} \sim \pi(\cdot \mid \boldsymbol{o})}\left[Q^{\pi}(\boldsymbol{s}, \boldsymbol{a})\right]$), and define the global
% advantage
% \begin{equation}
%      A(\boldsymbol o,\boldsymbol a)
%   = Q^{\pi_{\boldsymbol{\theta}}}(\boldsymbol o,\boldsymbol a)
%     - V_{\mathrm{tot}}(\boldsymbol o).
% \end{equation}

\begin{lemma} \label{lemma1}
% Under the standard assumptions for actor-critic methods
% (e.g., \cite{sutton1999policy,konda1999actor}), the updates implemented in
% Algorithm~\ref{VDAC} converge to a stationary point of $J(\theta)$, i.e.,
% \begin{equation}
%     \liminf_{k\to\infty} \|\nabla_\theta J(\theta_k)\| = 0
% \end{equation}
%
If (i) $\mathbb{E}[Q_{\textrm{G}}(\boldsymbol{o}_t, \boldsymbol{a}_t)] = Q^{\pi_{\boldsymbol{\theta}}}(\boldsymbol o,\boldsymbol a)$, (ii) learning rate sequences $\{\lambda_{a,k}\}$ and $\{\lambda_{c,k}\}$ satisfy the Robbins--Monro conditions:
    $\sum_{k}\lambda_{a,k}=\sum_{k}\lambda_{c,k}=\infty,
      \sum_{k}\lambda_{a,k}^2<\infty,
      \sum_{k}\lambda_{c,k}^2<\infty,
      \lambda_{a,k}/{\lambda_{c,k}}\rightarrow 0$, and (iii) the update of the proposed MARL can be considered as the update of a single-agent actor–critic policy-gradient update, our proposed VDAC-DNC algorithm can converge to a locally optimal policy~\cite{Foerste2018}. 
\end{lemma}
%proof：为什么满足这两个条件
%考虑DNC对收敛的影响
\begin{proof}
To prove the convergence of the proposed VDAC-DNC algorithm, we need to prove that the proposed VDAC-DNC algorithm satisfies conditions (i), (ii), and (iii).
First, the error introduced by the $Q$ function approximation
decreases as the neural networks are trained over time. Meanwhile, the value decomposition method is used to generate the global $Q$ value, ensuring that the difference between the actual
global $Q$-function and its approximated value remains bounded. 
Thus, the proposed VDAC-DNC algorithm satisfies condition (i) when the algorithm converges.
Note that the DNC module will not affect the convergence of the designed method since the DNC aims to find a better action from the original action space instead of creating new actions.   
%converts the continuous action of the BS into the discrete action that is actually executed.
%can be regarded as part of the environment dynamics since it maps the continuous action of the BS to the discrete action.
%During training, the environment trajectories are generated after this DNC conversion and used to train the global $Q_\textrm{G}$.
%affect the global $Q^{\pi_{\boldsymbol{\theta}}}$.
% Hence, the state transitions and rewards observed during training already reflect this combined dynamics.
% Since the global critic $Q_{\textrm G}$ is trained using TD targets computed from trajectories generated under the same combined dynamics, it is driven to satisfy the Bellman consistency associated with $\pi_{\boldsymbol{\theta}}$ in this modified environment.
%However, the critic $Q_\textrm{G}$ is trained exactly on trajectories generated under this combined dynamics. 
%Since the global $Q_\textrm{G}$ is trained on these
%Therefore, $Q_\textrm{G}$ is still approximating the correct $Q^{\pi_{\boldsymbol{\theta}}}$ for the policy together with the DNC module, and the condition (i) continues to hold.

Next, we prove that the proposed method meets condition~(ii). 
Here, the designed method can meet condition~(ii) via adjusting the value of the learning rate $\lambda_{a_k}$ and $\lambda_{c_k}$.
%since the learning rate sequences are designed parameters and can be constructed.
For example, let $\lambda_{a_k} = a / (k+1)^{\alpha}$ and $\lambda_{c_k} = c / (k+1)^{\beta}$ with $1/2 < \beta < \alpha \leq 1$.
% Using the $p$-series test, $\sum_{k\ge 0}(k+1)^{-p}$ diverges if $p\le 1$ and
% converges if $p>1$.
Then, the first and second conditions in (ii) can be satisfied as
% \begin{align}
%     &  \\
% &\sum_{k=0}^\infty \lambda_{c,k}
% = c\sum_{k=0}^\infty (k+1)^{-\beta}=\infty \iff \beta\le 1
% \end{align}
\begin{equation}
    \sum_{k=0}^\infty \lambda_{a,k}
= a\sum_{k=0}^\infty (k+1)^{-\alpha}=\infty, \quad \text{when} \quad \alpha\le 1,
\end{equation}
%Similarly, the second condition in (ii) can be proved as:
\begin{equation}
    \sum_{k=0}^\infty \lambda_{a,k}^2
= a^2\sum_{k=0}^\infty (k+1)^{-2\alpha}<\infty, \quad \text{when} \quad \alpha>\tfrac12,
\end{equation}
With regards to the third condition in (ii), since $\frac{1}{2}<\beta\leq1$, we have
\begin{equation}
    \frac{\lambda_{a,k}}{\lambda_{c,k}}
=\frac{a}{c}(k+1)^{\beta-\alpha}
\rightarrow 0, \quad \text{when} \quad \alpha>\beta.
\end{equation}
Thus, the proposed VDAC-DNC algorithm can satisfy condition~(ii) by setting specific learning rates for the actor and critic networks.

To prove that the proposed method meets condition (iii), we need to rewrite the expected gradient as a standard single-agent actor-critic policy gradient.
Given~\eqref{return} and the policy-gradient theorem for discounted Markov decision processes (MDPs)
\cite{sutton1999policy,10993496}, we have
\begin{equation}
      \nabla_{\boldsymbol{\theta}} J(\boldsymbol{\theta})
  = \mathbb{E}_{\pi_{\boldsymbol{\theta}}}\Big[
      \sum_{i\in\{\textrm{B}\}\cup\mathcal V}
      \nabla_{\boldsymbol{\phi}_i}
      \log \pi_{\boldsymbol{\phi}_i}(\boldsymbol a_{i,t}|\boldsymbol o_{i,t})\,
      Q^{\pi_{\boldsymbol{\theta}}}(\boldsymbol o_t,\boldsymbol a_t) 
    \Big], \label{eq:ori}
  %\mathbb{E}_{\pi_{\boldsymbol{\theta}}}\big[\sum_{i\in\{\textrm{B}\}\cup\mathcal V}
    %   \nabla_{\boldsymbol{\theta}} \log \pi_{\boldsymbol{\theta}}(\boldsymbol a_t|\boldsymbol o_t)\,
    %   Q^{\pi_{\boldsymbol{\theta}}}(\boldsymbol o_t,\boldsymbol a_t)
    % \big]. 
\end{equation}
where $\boldsymbol{\phi}_i$ denotes the local actor parameters of agent $i$
(BS or user).
Then, the expected gradient of $J(\boldsymbol{\theta})$ can be written as
% \begin{align}
%     \nabla_{\boldsymbol{\theta}} J(\boldsymbol{\theta})
%   & = \mathbb{E}_{\pi_{\boldsymbol{\theta}}}\Big[
%       \sum_{i\in\{\textrm{B}\}\cup\mathcal V}
%       \nabla_{\boldsymbol{\varphi}_i}
%       \log \pi_{\boldsymbol{\varphi}_i}(\boldsymbol a_{i,t}|\boldsymbol o_{i,t})\,
%       Q^{\pi_{\boldsymbol{\theta}}}(\boldsymbol o_t,\boldsymbol a_t)
%     \Big] \notag \\
%     & =\mathbb{E}_{\pi_{\boldsymbol{\theta}}}\Big[
%       \nabla_{\boldsymbol{\varphi}_i}
%       \log \prod_{i\in\{\textrm{B}\}\cup\mathcal V}\pi_{\boldsymbol{\varphi}_i}(\boldsymbol a_{i,t}|\boldsymbol o_{i,t})\,
%       Q^{\pi_{\boldsymbol{\theta}}}(\boldsymbol o_t,\boldsymbol a_t)
%     \Big].
%   \label{eq:pg-q-short}
% \end{align}
\begin{equation}
    \nabla_{\boldsymbol{\theta}} J(\boldsymbol{\theta}) =\mathbb{E}_{\pi_{\boldsymbol{\theta}}}\Big[
      \nabla_{\boldsymbol{\phi}_i}
      \log \prod_{i\in\{\textrm{B}\}\cup\mathcal V}\pi_{\boldsymbol{\phi}_i}(\boldsymbol a_{i,t}|\boldsymbol o_{i,t})\,
      Q^{\pi_{\boldsymbol{\theta}}}(\boldsymbol o_t,\boldsymbol a_t)
    \Big]. \label{eq:grident}
\end{equation}
By using the factorization \eqref{eq:joint-policy}, the~\eqref{eq:grident} can be written as
\begin{equation}
    \nabla_{\boldsymbol{\theta}} J({\boldsymbol{\theta}}) = \mathbb{E}_{\pi_{\boldsymbol{\theta}}}\Big[
      \nabla_{\boldsymbol{\theta}} \log \pi_{\boldsymbol{\theta}}(\boldsymbol a_t|\boldsymbol o_t)\,
      Q^{\pi_{\boldsymbol{\theta}}}(\boldsymbol o_t,\boldsymbol a_t)
    \Big]. \label{eq:finalgd}
\end{equation}
% where $\boldsymbol{\pi}(\boldsymbol{a}_t|\boldsymbol{o}_t) = \prod_{i\in\{\textrm{B}\}\cup\mathcal V}\pi_{\boldsymbol{\phi}_i}(\boldsymbol a_{i,t}|\boldsymbol o_{i,t})$.
From~\eqref{eq:finalgd}, we see that the expected gradient can be considered as the update of a single-agent actor-critic policy. Consequently, the proposed VDAC-DNC algorithm satisfies condition (iii).
%Note that the DNC module does not change the policy gradient structure thus the condition (ii) continues to hold.
% Finally, the proposed DNC module is applied as a processing operator on the BS's continuous action before interacting with the environment. From the perspective of the actor–critic update, the DNC module can be regarded as part of the environment dynamics since it only affect the global $Q^{\pi_{\boldsymbol{\theta}}}$ and does not change the policy gradient structure.
% %Therefore, all assumptions in our convergence analysis remain valid for the proposed VDAC-DNC algorithm.
This completes the proof. Our simulation result in Section~\ref{se:simulation}
also demonstrates that the proposed VDAC-DNC algorithm
effectively converges.
\end{proof}
% From Lemma~\ref{lemma1}, we see that the proposed VDAC-DNC algorithm is guaranteed to converge to a locally optimal solutions of problem~\ref{obj}.
% Our simulation result in Section~\ref{se:simulation} also demonstrates that the
% proposed VDAC-DNC algorithm effectively converges.
%convergence of the proposed VDAC-DNC algorithm depends on the gap between the actual global $Q$ function and the global critic network $Q$ function and the form of expected gradient. 
%In our algorithm, all conditions are satisfied.

%Given~\eqref{eq:pg-q-short}, 
%Then, 
% According to Lemma~\ref{lemma1}, an actor-critic that follows this gradient converges to a local maximum of the expected return, subject to assumptions in~\cite{sutton1999policy}.
%Thus, the proposed VDAC-DNC algorithm converges to a locally optimal policy.

 %\end{theorem}

\subsection{Implementation Analysis}
The implementation of the VDAC-DNC scheme for image semantic transmission consists of an offline training stage and an online implementation stage.
In the offline training stage, the BS first allocates the sub-channels for each user using the designed DNC method as shown in Algorithm~\ref{DNC} and sends the sub-images to users.
Then it requires the 1) sub-image transmission status of each user, 2) CSI and the transmission rate of each channel to calculate local $Q_{\boldsymbol{\psi}^{\textrm{B}}}(\boldsymbol{o}_{v,t}^{\textrm{B}}, \boldsymbol{a}_{v,t}^{\textrm{B}})$ for training its critic network.
% to train the critic network of the BS, the BS requires the 1) sub-image transmission status of each user, 2) CSI and the transmission rate of each channel.
% Given this information, the BS allocates the sub-channels using the designed DNC method as shown in Algorithm~\ref{DNC}.
Once users receive the sub-images sent by the BS, each user requires 1) the current CSI, 2) the received signal, and 3) the current remaining power budget to calculate the local $Q$ values for training its critic network.
Then, each user transmits the local $Q_{\boldsymbol{\psi}_{v}^{\textrm{U}}}(\boldsymbol{o}_{v,t}^{\textrm{U}}, \boldsymbol{a}_{v,t}^{\textrm{U}})$ to the BS for the calculation of the global $Q$ function $Q_{\textrm{G}}(\boldsymbol{o}_t, \boldsymbol{a}_t)$. In particular, the BS aggregates local $Q$ values $Q_{\boldsymbol{\psi}_{v}^{\textrm{U}}}(\boldsymbol{o}_{v,t}^{\textrm{U}}, \boldsymbol{a}_{v,t}^{\textrm{U}})$ of all users and the local $Q_{\boldsymbol{\psi}^{\textrm{B}}}(\boldsymbol{o}_{v,t}^{\textrm{B}}, \boldsymbol{a}_{v,t}^{\textrm{B}})$ of the BS using \eqref{sumgq} or \eqref{qmix}.
Based on the state and action information, as well as the global $Q$ value, the parameters of the critic network and actor network of the BS and users are updated according to \eqref{bscri}--\eqref{eq:user-actor-update-short}.
\nbl{Note that, in the training procedure of the proposed method, we update only the actor and critic networks. We do not need to train DNC since it is an action-search module and does not contain trainable neural network parameters.}
% \nbl{Note that DNC is executed at each RL step after the BS actor generates a continuous action and does not contain trainable neural-network parameters, thus it does not require a separate training process.}
% In the online implementation stage, the feedback selection and power control are independently determined by each user. 
In the online stage, since the designed method is trained by an offline training method, it does not need any data to update the neural network models, which significantly reduces the implementation complexity. 
\nbl{The designed method can be deployed in practical
scenarios since 1) the proposed framework explicitly considers practical wireless constraints and 2) the information required by the proposed method can be easily obtained
in practical systems.}

\subsection{Complexity Analysis}
The computational complexity of the proposed VDAC-DNC scheme consists of three main components: 1) the training of actor network, critic network, 2) the DNC process and 3) the training of global $Q$ function at the BS.

First, we analyze the complexity of training the actor and critic network.
% We assume the computation of the data or gradient transmission in the MLP has a complexity of $\mathcal{O}(N)$, where $N$ is the number of parameters in the MLP.
% In particular, the BS and each user maintain local Q-networks $Q_{\boldsymbol{\psi}^{\textrm{B}}}$ and 
% $Q_{\boldsymbol{\psi}_v^{\textrm{U}}}$, actor networks $\boldsymbol{\pi}_{\boldsymbol{\phi}^{\textrm{B}}}$ and $\boldsymbol{\pi}_{\boldsymbol{\phi}_v^{\textrm{U}}}$. 
% Each network is approximated by a deep neural network (DNN).
Let $N_Q^{\textrm{B}}, N_Q^{\textrm{U}}, N_{\pi}^{\textrm{B}}, N_{\pi}^{\textrm{U}}$ be the numbers of parameters of the actor and critic networks. 
Then, due to the parallel training across the BS and the users, the complexity of training the actor and critic networks per iteration is $\mathcal{O}\big(\max\left\{N_{\pi}^{\textrm{B}} + N_Q^{\textrm{B}},  N_Q^{\textrm{U}}  + N_{\pi}^{\textrm{U}}\right\}\big)$, where $\max\{\cdot \}$ implies that the training complexity of the actor and critic networks depends on the training complexity $N_{\pi}^{\textrm{B}} + N_Q^{\textrm{B}}$ and the training complexity $N_Q^{\textrm{U}} + N_{\pi}^{\textrm{U}}$ at a user~\cite{11216397,11270936}.
%indicates that per-iteration trainig time is dominated by the slower branch, i.e., the larger computational cost between the BS update and the user update.
%取决于BS还是用户的actor critic network哪个训练的久
%Let $V$ be the number of users.
% For a feedforward DNN with $N$ trainable parameters, both the forward 
% propagation and the backpropagation steps have computational complexity 
% $\mathcal{O}(N)$. In the proposed framework, 

%\paragraph{Complexity of DNC.}
Then we analyze the computational complexity of the DNC process.
As shown in Algorithm~\ref{DNC}, the BS checks the Q-values of neighboring actions using its critic networks in each iteration. 
Thus, the complexity of the DNC per iteration is $\mathcal{O}\left( N_Q^{\textrm{B}}\right)$.
Let $y_0$ and $\epsilon_0$ be the initial values of $y$ and $\epsilon$ in
Algorithm~\ref{DNC}, and let $c_k$ and $c_\beta$ be the corresponding
decrement steps. Then the total number of DNC iterations is bounded by~\cite{bertsimas1993simulated}
\begin{equation}
    K_{\textrm{D}}
  \le
  \max\!\left\{
    \left\lceil \frac{y_0}{c_k} \right\rceil,
    \left\lceil \frac{\epsilon_0}{c_\beta} \right\rceil
  \right\},
\end{equation}
where $\left\lceil \cdot \right\rceil$ is the ceiling function.
Consequently, the complexity of the DNC process at an RL iteration is
$\mathcal{O}\!\left(   K_{\textrm{D}}  N_Q^{\textrm{B}} \right)$.
%\paragraph{Complexity of the VDAC--DNC algorithm.}
% In the main VDAC--DNC procedure (Algorithm~\ref{VDAC}), at each time slot the
% following neural networks are updated once (one forward and one backward
% pass):

% \begin{itemize}
%   \item BS actor $\boldsymbol{\pi}_{\boldsymbol{\phi}^{\textrm{B}}}$:
%         complexity $\mathcal{O}(V N_{\pi}^{\textrm{B}})$;
%   \item BS Q-network $Q_{\boldsymbol{\psi}^{\textrm{B}}}$:
%         complexity $\mathcal{O}(V N_Q^{\textrm{B}})$;
%   \item User actors $\{\boldsymbol{\pi}_{\boldsymbol{\phi}_v^{\textrm{U}}}\}$:
%         complexity $\mathcal{O}(V N_{\pi}^{\textrm{U}})$;
%   \item User Q-networks $\{Q_{\boldsymbol{\psi}_v^{\textrm{U}}}\}$:
%         complexity $\mathcal{O}(V N_Q^{\textrm{U}})$;
%   \item Mixing hypernetwork $\boldsymbol{\Gamma}_{\boldsymbol{\Phi}}$:
%         complexity $\mathcal{O}(N_{\textrm{mix}})$.
% \end{itemize}

With regards to the training of the global $Q$ function, if we use \eqref{sumgq} to approximate the global $Q$ function, the computational complexity is $\mathcal{O}(V)$ since we only need to summarize the $Q$ functions of all agents.
On the other hand, if we use \eqref{qmix} to approximate the global $Q$ function,
the computational complexity of the global $Q$ function depends on the MLP architecture. Let $N_{Q}^{\textrm{M}}$ be the number of parameters in the MLP. The training complexity of the global $Q$ function is $\mathcal{O}(N_{Q}^{\textrm{M}})$.
% The complexity of training 

Given the aforementioned analysis results and assuming we use \eqref{qmix} to approximate the global $Q$ function, the overall complexity of VDAC-DNC per iteration is
\begin{equation}
  \mathcal{O}\Big(
     K_{\textrm{D}}  N_Q^{\textrm{B}} + 
    \textrm{max}\left\{N_{\pi}^{\textrm{B}} + N_Q^{\textrm{B}} ,  N_Q^{\textrm{U}}  + N_{\pi}^{\textrm{U}}\right\} + N_Q^{\textrm{M}}
  \Big).
  \label{eq:complexity-per-step}
\end{equation}

\section{Simulation Results}\label{se:simulation}
For our simulations, we consider a network consisting of one BS and $V=15$ users.
The distance between each user and the BS is randomly selected within the range of $[8,50]$ m, and the downlink interference $I^{\textrm{D}}_{n,v}$ is randomly selected within the range of $[10^{-9},10^{-8}]$ W.
Other parameters are summarized in Table~\ref{tab:params}.
\nbl{To validate the performance of our proposed scheme, we consider three benchmarks: 1) VDAC (w/o DNC) that optimizes sub-image allocation and feedback selection using the designed RL method but does not use DNC to find better actions per iteration\nbl{~\cite{JMLR:v21:20-081}}, 2) a DNC based MAQAC algorithm that consists of a centralized critic and decentralized actors\nbl{~\cite{lowe2017multi}}, and 3) DNC-based independent Q algorithm in which each user utilizes a deep $Q$ network (DQN) to optimize its transmit power and feedback selection without considering the actions of other users\nbl{~\cite{tan1993multi}}. 
%The latter two benchmarks use the same DNC configuration in our simulations for fair comparisons.
Both the proposed VDAC-DNC and VDAC (w/o DNC) methods use a MLP to approximate the global Q function~\cite{JMLR:v21:20-081}.
\nbl{Baseline 1) is used to show the benefits achieved by the designed DNC}.
Baseline 2) is used to compare value-decomposition training with a centralized critic compared to a distributed critic.
Baseline 3) is used to show the benefit of collaborative MARL over independent trained agents.}
% \nbl{Unless otherwise specified, the simulation parameters are set based on~\cite{yan2022semanticRA}: the bandwidth is $B=1.8\times10^{5}$ Hz, the noise power is $N_0=10^{-9}$ W, the BS downlink transmit power is $P^{\textrm{D}}=1$ W, and each user is assigned $N=6$ sub-channels.
% Moreover, the outage probability is $\xi=0.2$, each image must be transmitted within $T=6$ time slots, the neighboring-action search range of DNC is $l=10$, the number of reported content categories is $D=4$, $\varepsilon=1$, and the available feedback power of each user is $p_{\textrm{max}}=5$ W.}
\begin{table}[t] 
\caption{Simulation parameter\cite{yan2022semanticRA}.}
\centering
\renewcommand{\arraystretch}{1.2}
\begin{tabular}{|c|c|c|c|}
\hline
\textbf{Parameters} & \textbf{Values} & \textbf{Parameters} & \textbf{Values} \\
\hline
$B$                 & $1.8\times10^{5}\ \mathrm{Hz}$ & $N_0$ & $1\times 10^{-9}\ \mathrm{W}$ \\
\hline
$P^{\textrm{D}}$                 & $1 \mathrm{W}$          & $N$       & $6$ \\
\hline
$N_0$                 & $10^{-9} \mathrm{W}$                           & $\xi$       & $0.2 $ \\
\hline
$T$ & $6$           & $l$ & $10$ \\
\hline
$V$          & $15$              & $\varepsilon$ & $1$ \\
\hline
$D$       & $4$                 & $p_{\textrm{max}}$ & $5 \mathrm{W}$ \\
\hline
\end{tabular}
\label{tab:params}
\end{table}
% \begin{table}[t]
% \caption{Key neural network hyperparameter settings~\cite{akkerman2024dynamic}}
% \centering
% %\renewcommand{\arraystretch}{1.2}
% \begin{tabular}{|c|c|}
% \hline
% \textbf{Parameters} & \textbf{Values} \\
% \hline
% Optim & SGD \\
% \hline
% $\lambda_a$ & $1\times10^{-4}$ \\
% \hline
% $\lambda_c$ & $1\times10^{-3}$ \\
% \hline
% $\gamma$ & $0.99$ \\
% \hline
% The size of a hidden layer at MLP & 32 \\
% \hline
% Activation function of actor & Tanh\\
% \hline
% Update frequency of actor and critic & 40 \\
% \hline
% Hidden layer size of QMix MLP & 32 \\
% \hline
% \texttt{clipping\_factor} & $0.2$ \\
% \hline
% Fourier order of state encoding& 3 \\
% \hline
% Encoding batch size of VAE & 16 \\
% \hline
% Latent channel of VAE & 4 \\
% \hline
% % \texttt{hiddenActorLayerSize} & $16$ \\
% % \hline
% \end{tabular}
% \label{tab:nn_params_key_2col}
% \end{table}

\begin{figure}[t]
    \centering
    \includegraphics[width=8cm]{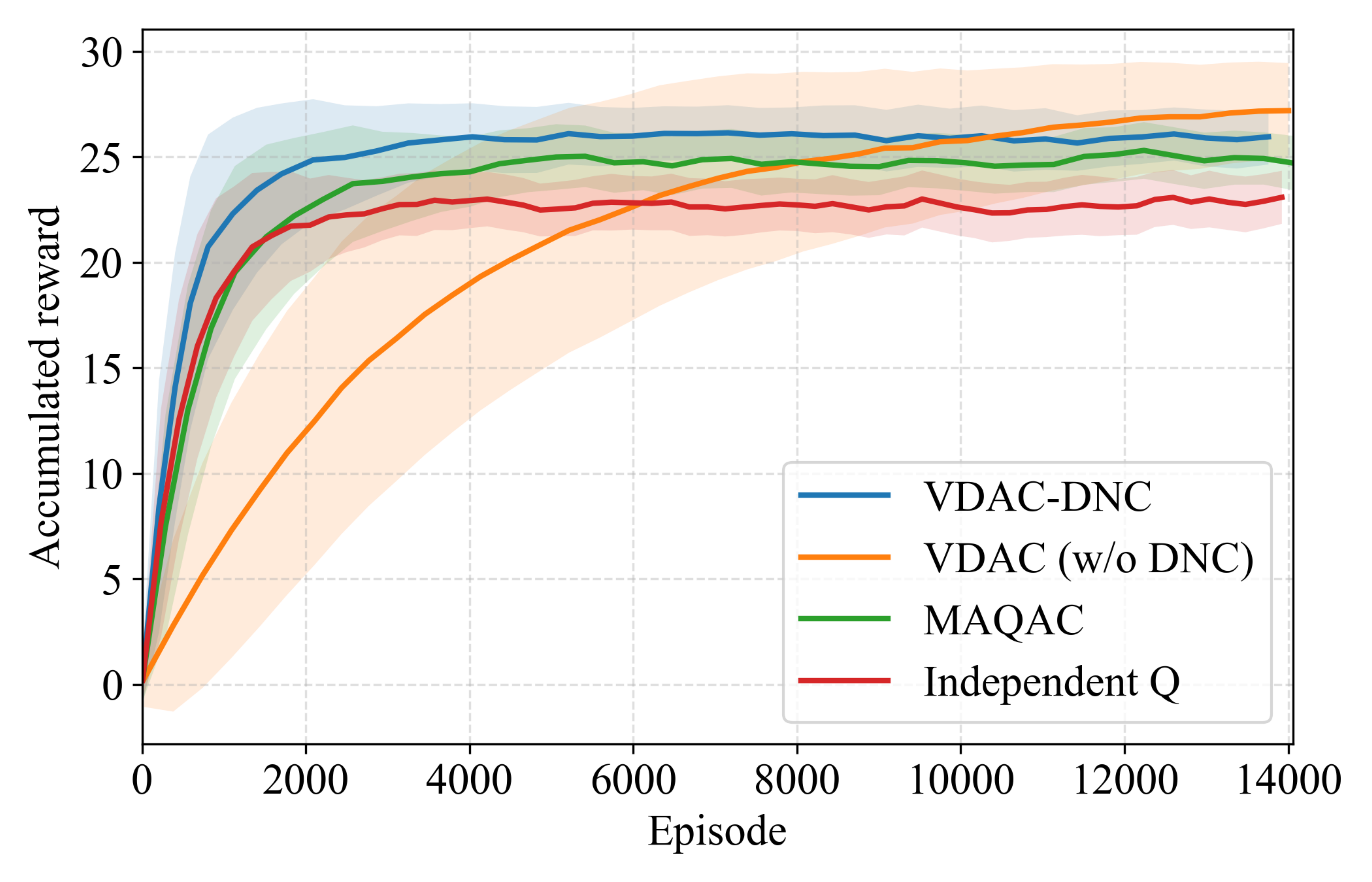}
    \vspace{-0.2cm}
    \nblcaption{Convergence of the considered algorithm}
    \label{fig:reward}
    \vspace{-0.4cm}
\end{figure}

% \begin{figure}[t]
%     \centering
%     \includegraphics[width=8cm]{Journal_Figure/reward_vs_power.pdf}
%         \vspace{-0.3cm}
%     \caption{Accumulated reward versus the total power of a user}
%     \label{dif-power}
%       \vspace{-0.6cm}
% \end{figure}

% This section presents numerical results validating the enhanced image transmission and performance of our proposed schemes.

% 2是为了展示分布式和集中式的
% 3是为了展示分布式多智体结构

% 分布式和集中式
% 分布式
% 
% Critic 是global Q训练的，选择动作使得global Q
% Critic 是independent Q 训练的，选择动作使得自己的Q增大的动作

Fig.~\ref{fig:reward} shows the convergence of the considered algorithms during the training process. The shaded area represents half of the standard deviation of the accumulated training reward. 
From Fig.~\ref{fig:reward}, we see that as the number of training episodes increases, the rewards of all considered schemes increase first and then stabilize.
In particular, VDAC-DNC improves the reward by up to 5.04\% and 18.55\% compared to the MAQAC and independent Q when the number of episodes is 14000. 
The 5.04\% gain stems from the fact that the proposed method used a MLP to approximate the global function using local Q functions of the users and BS, while the MAQAC uses a high-dimensional centralized critic to approximate a global Q function, causing overfitting.
%Hence, the proposed method can make the training of global Q function easier.
The 18.55\% gain stems from the fact that the proposed method approximates a global Q function to update the local Q function of each agent, thus allowing the users and the BS to collaboratively select feedback and sub-image transmission to maximize their team reward in \eqref{team reward} instead of their individual rewards.

\begin{figure}[t]
    \centering
    \includegraphics[width=8cm]{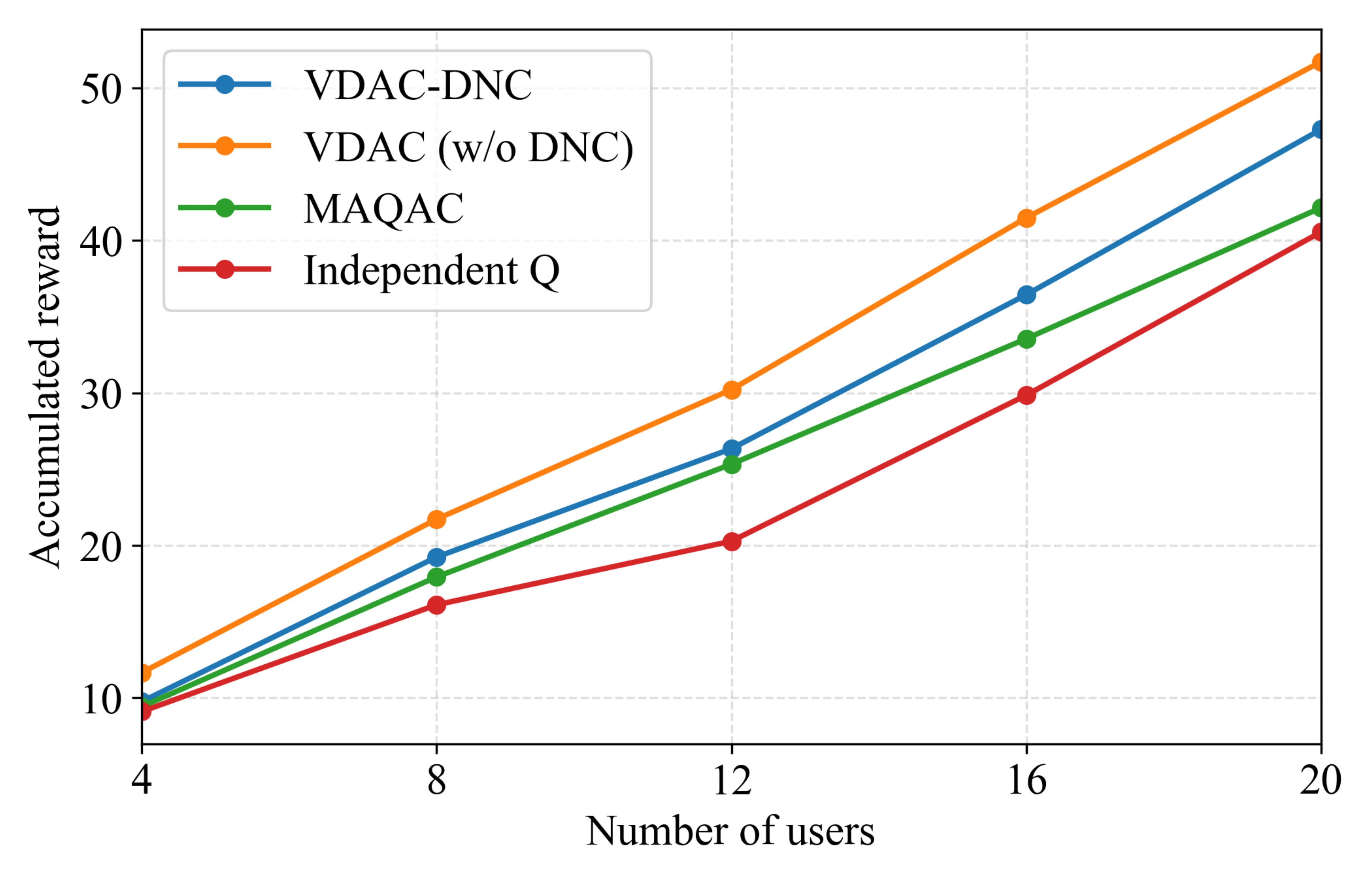}
    \vspace{-0.2cm}
    \caption{Accumulated reward versus the number of users}
    \label{dif-user}
    \vspace{-0.4cm}
\end{figure}

Fig.~\ref{dif-user} shows how the accumulated rewards resulting from the considered methods at convergence vary when the number of users ranges from 4 to 20.
In this figure, we see that the rewards of all the schemes increase as the number of users increases, since the number of images to be transmitted increases.
Fig.~\ref{dif-user} also shows that the gap in terms of the reward between the proposed VDAC-DNC scheme and the MAQAC increases as the number of users increases.
This is because our designed method uses a MLP to nonlinearly approximate the global Q function using the local Q function of the BS and users, which provides a structured and scalable approximation that remains effective as the number of users grows.
In contrast, the MAQAC scheme relies on a high-dimensional centralized critic taking into account the states and actions of all the agents to directly approximate the global Q-function. As a result, this critic becomes harder to train and more prone to overfitting as the number of users increases.
% the proposed VDAC-DNC scheme can enhance the performance by up to 12.17\% and 16.57\% compared to the MAQAC scheme and the independent Q scheme when the number of users is 20.
% The 12.17\% gain is because our designed method uses a MLP to nonlinearly approximate the global Q function using the local Q function of the BS and users, thus avoiding overfitting by using a high-dimensional centralized critic to approximate a global Q function.
% The 16.57\% gain stems from the fact that the proposed method enables users to cooperatively determine when to transmit the feedback to avoid a number of feedback collapses in one time slot, while the users under the independent Q method transmit the feedback without considering the actions of other users, which causes the feedback transmission to collapse and fail with high probability.
%BS and the users to collaboratively determine their actions in order to optimize their team reward (i.e., total MSE of all the images) instead of their individual rewards.
Fig.~\ref{dif-user} also shows that the proposed VDAC-DNC scheme can enhance the performance by up to 16.57\% compared to the independent Q scheme when the number of users is 20.
The 16.57\% gain stems from the fact that the proposed method allows users to coordinate their feedback transmission, preventing a large number of users from sending feedback in the same time slot thus reducing feedback transmission speed. 

% \begin{table}[t]
% \centering
% \caption{Variance of feedback for different schemes and user numbers.}
% \begin{tabular}{|l|c|c|c|c|c|}
% \hline
% Scheme / Users & 4 users & 8 users & 12 users & 16 users & 20 users \\
% \hline
% Independent Q 
% & 3.14 & 5.92 & 2.89 & 6.14 & 8.14 \\
% \hline
% QAC
% & 1.47 & 1.47 & 14.33 & 8.22 & 8.47 \\
% \hline
% VDAC (w/o DNC) 
% & 2.92 & 6.47 & 2.92 & \textbf{1.58} & 2.81 \\
% \hline
% VDAC-DNC
% & \textbf{0.22} & \textbf{0.14} & \textbf{0.56} & 3.89 & \textbf{1.47} \\
% \hline
% \end{tabular}
% \label{t:var-user}
% \end{table}
%\begin{table}[t]
% \centering
% \caption{Number of feedback for different schemes in a episode.}
% \begin{tabular}{|l|c|c|c|c|c|c|c|}
% \hline
% Scheme / Time & 1 & 2 & 3 & 4 & 5 & 6 & Variance \\
% \hline
% Independent Q 
% & 8 & 8 & 7 & 3 & 5 & 5 & 3.33 \\
% \hline
% QAC
% & 5 & 6 & 7 & 5 & 4 & 3 & 1.67 \\
% \hline
% VDAC (w/o DNC) 
% & 5 & 6 & 6 & 4 & 4 & 5 & 0.67 \\
% \hline
% VDAC-DNC
% & 5 & 5 & 5 & 5 & 5 & 6 & 0.14 \\
% \hline
% \end{tabular}
% \label{t:var-user}
% \end{table}

% \begin{table}[t]
% \centering
% \caption{Convergence time with and without DNC under different users settings.}
% \begin{tabular}{|l|c|c|c|c|c|}
% \hline
% Scheme / Users & 4  & 8  & 12  & 16  & 20  \\
% \hline

% VDAC w/o DNC (s) & 9586  & 17240 & 26224 & 33465 & 40768 \\ \hline
% VDAC-DNC (s)  & \textbf{2786}  & \textbf{3588}  & \textbf{6750}  & \textbf{8039}  & \textbf{10122} \\ \hline
% Ratio   & 3.44     & 4.80     & 3.89     & 4.16     & 4.03     \\ \hline
% \end{tabular}
% \label{t:dif-user}
% \end{table}
\begin{figure}
    \centering
    \includegraphics[width=8cm]{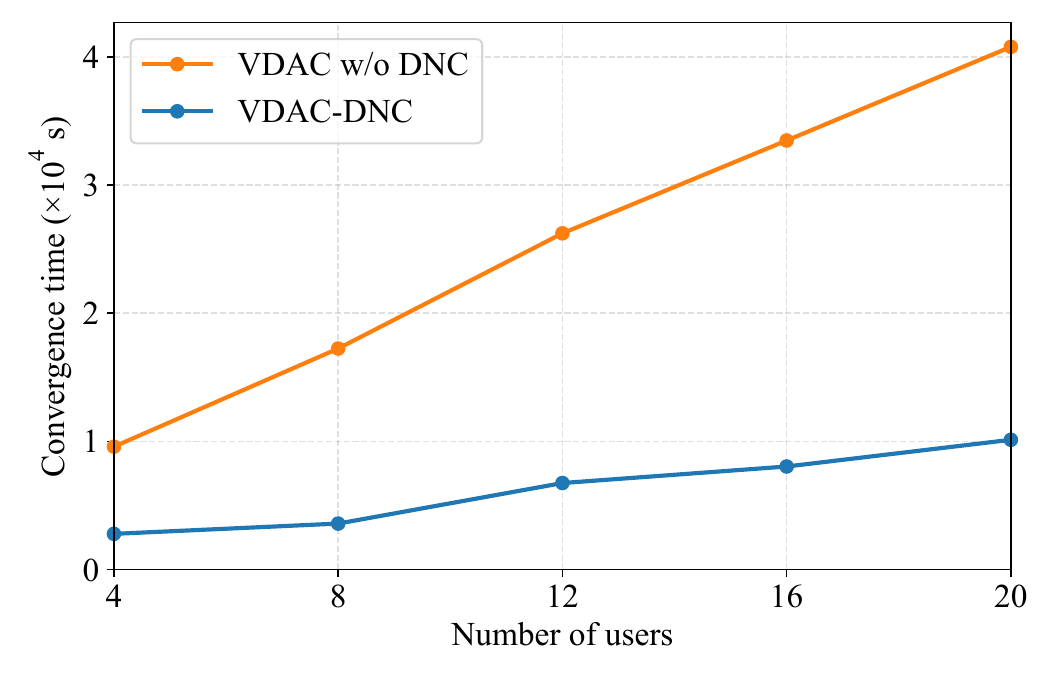}
    \vspace{-0.2cm}
    \caption{Convergence time with and without DNC under different user settings}
    \label{fig:dif-user-time}
    \vspace{-0.2cm}
\end{figure}

\begin{figure}[t]
    \centering
    \includegraphics[width=8cm]{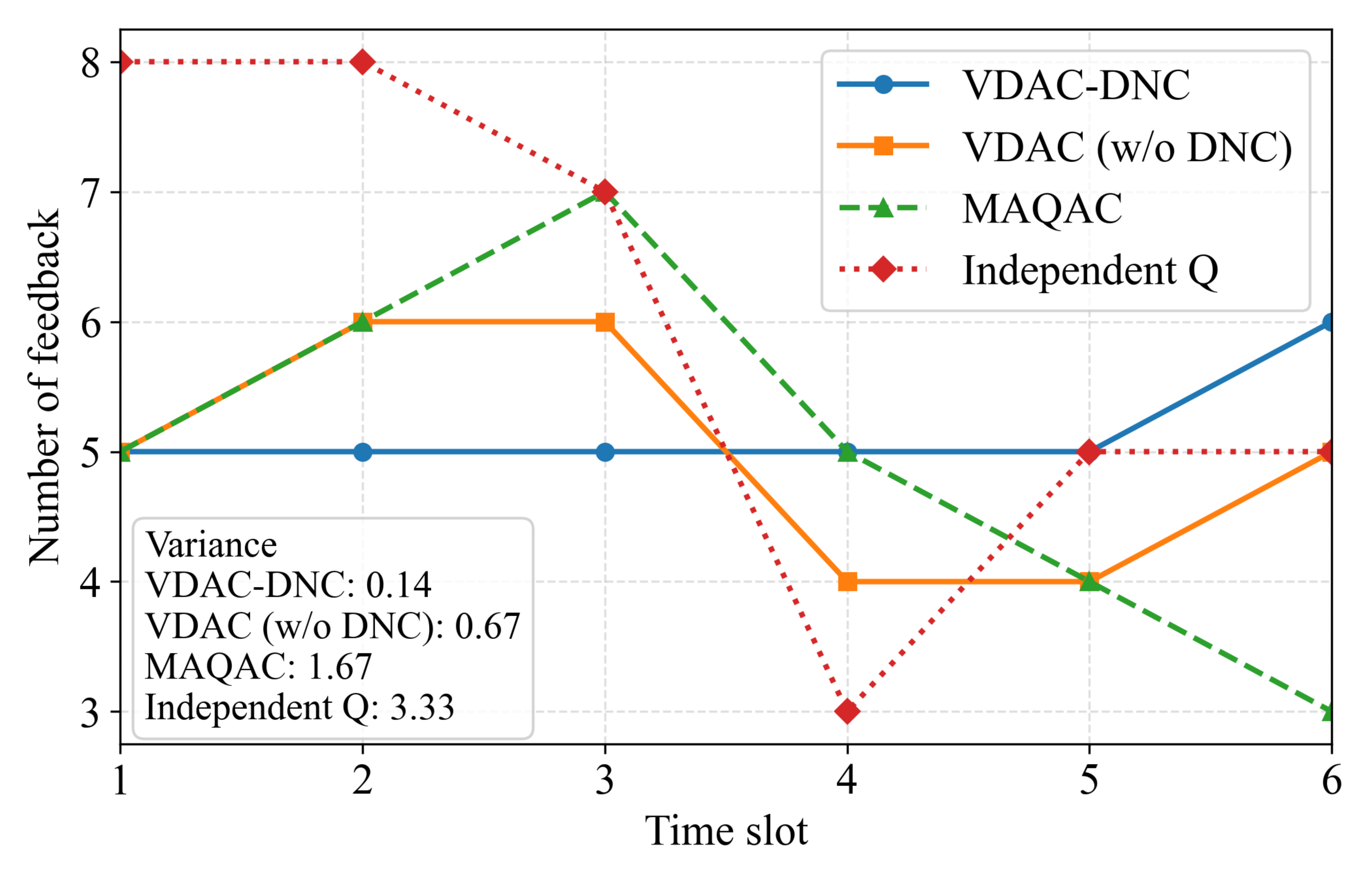}
    \vspace{-0.2cm}
    \caption{Number of feedback transmitted over an episode}
    \label{fig:fb}
    \vspace{-0.4cm}
\end{figure}

\begin{figure*}[t]
    \centering
    \includegraphics[width=15cm]{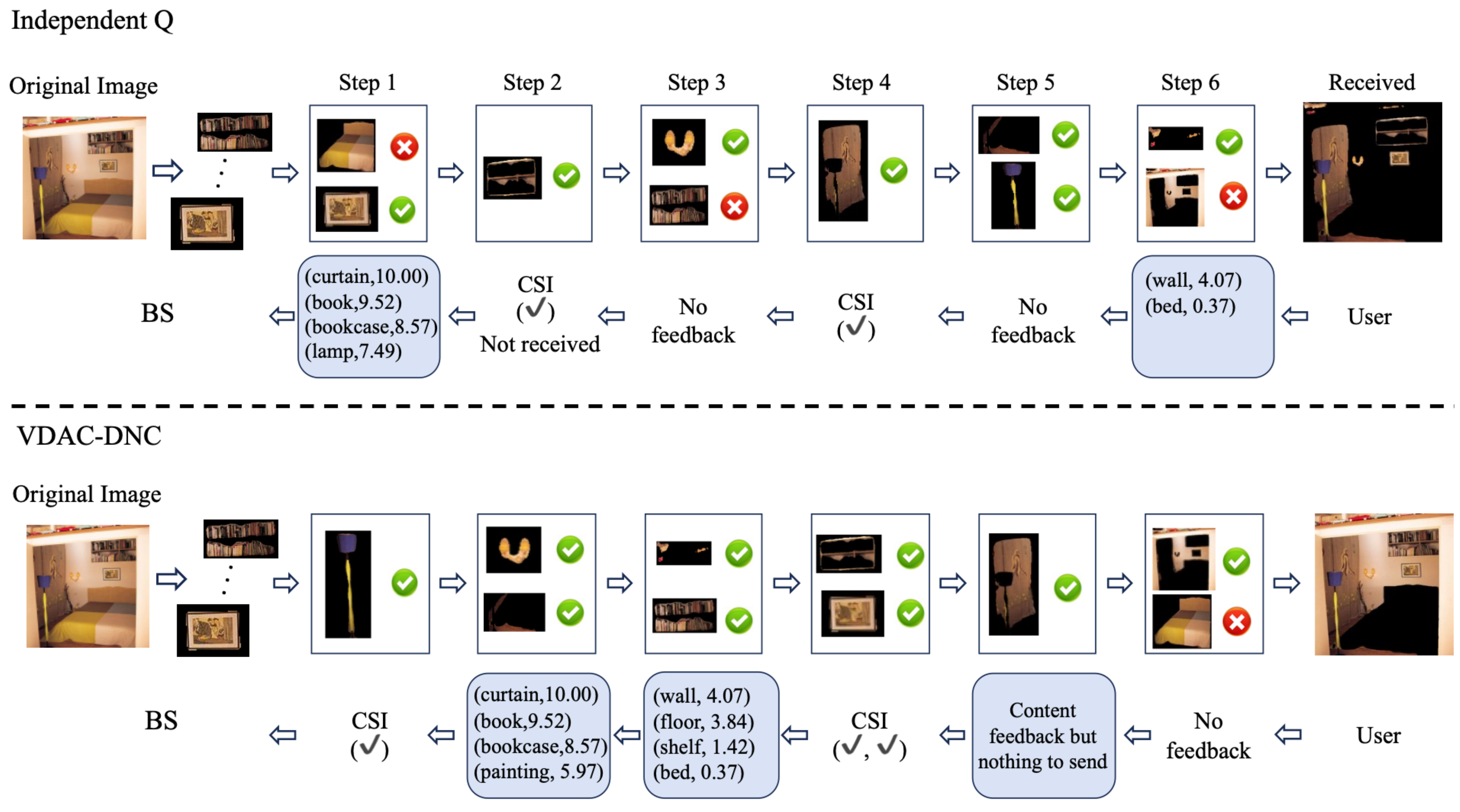}
    \vspace{-0.2cm}
    \caption{Visualization of sub-images and feedback messages transmission in each step during an episode}
    \label{fig:visualization}
    \vspace{-0.2cm}
\end{figure*}

\begin{figure}[t]
    \centering
    \includegraphics[width=8cm]{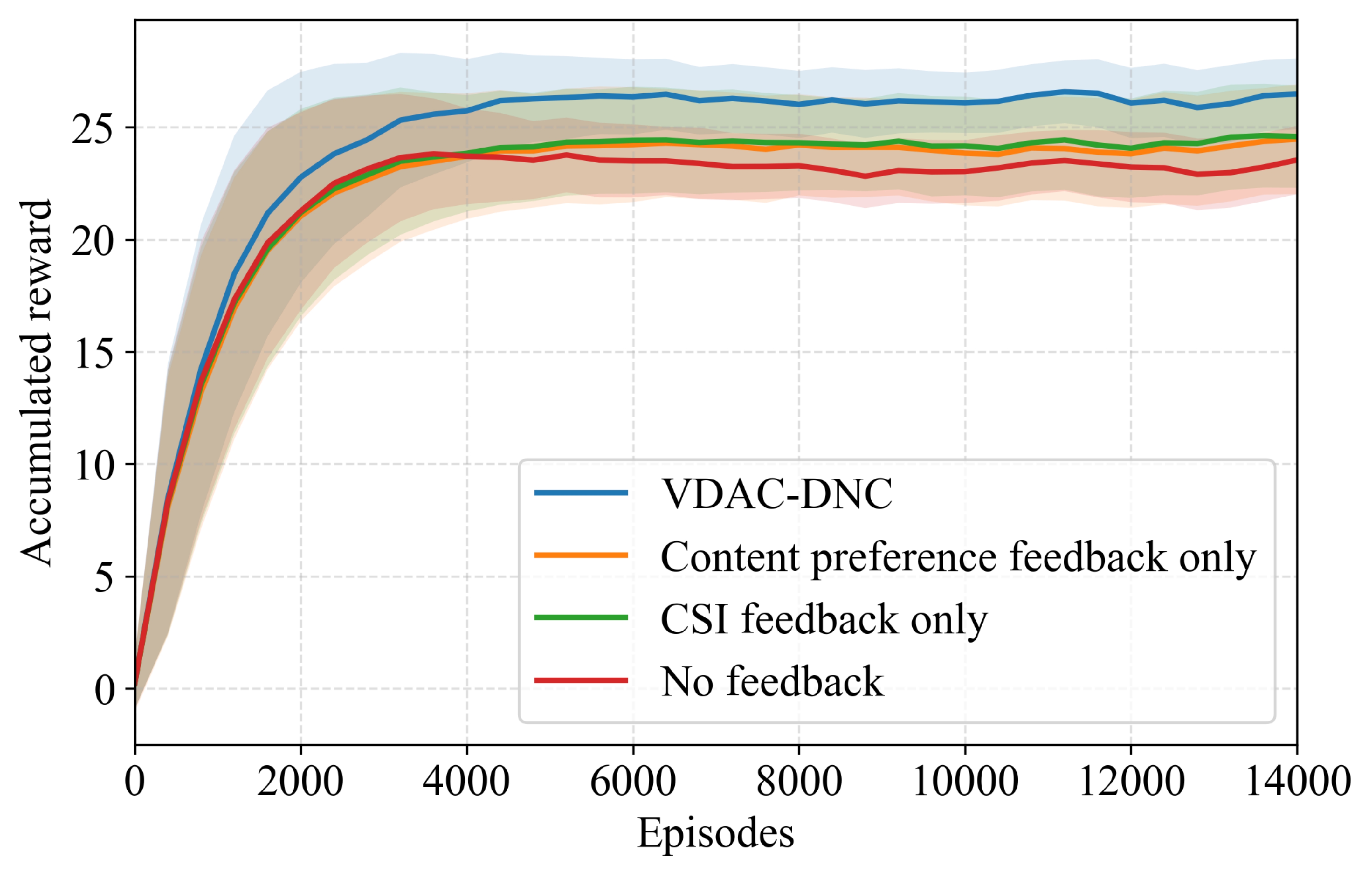}
    \vspace{-0.2cm}
    \caption{Convergence of the designed methods with different types of feedback}
    \label{fig:dif-fb}
    \vspace{-0.4cm}
\end{figure}

%in order to optimize their team reward (i.e., total MSE of all the images) instead of individual rewards.
%In addition, 
%We see that the proposed VDAC-DNC method is lower than that of the VDAC (w/o DNC).
%This is because the DNC find the better action across the neighboring action, resulting the action 

% Table~\ref{t:var-user} shows the variance of the number of feedback transmitted to the BS at each time step when the number of users ranges from 4 to 20.
% The variance represents fluctuations in the number of feedback sent by each user to the BS at each step, reflecting the degree of user cooperation.
% If the variance is high, it indicates that at certain time steps, a large number of users may concentrate their feedback transmissions, potentially causing channel congestion and transmission failures. Conversely, at other time steps, no users may send feedback, resulting in wasted channel resources.
% A high variance indicates that, in some time steps, many users transmit feedback simultaneously, which can congest the channel and increase the risk of transmission failures. In other time steps, few or no users transmit feedback, resulting in wasted channel resources.
% From this table, we see that the variance of the proposed scheme VDAC-DNC is much lower than the compared MAQAC scheme and the independent Q scheme, indicating the value decomposition framework play a key role in promoting cooperation among agents by aligning individual utilities with the shared team objective.

Fig.~\ref{fig:dif-user-time} shows the convergence time of the proposed scheme with DNC and without DNC under different user settings.
From this figure, we see that the convergence time of both schemes increases as the number of users increases.
This is because the joint action space grows significantly with the number of users, requiring the MARL to spend more time for exploration and convergence towards optimal actions.
Fig.~\ref{fig:dif-user-time} also shows that the proposed scheme with DNC improves the speed of convergence by up to 4.16x compared to the proposed scheme without DNC when the number of users is 16. This is because DNC selects an action with the maximum $Q$ value from the small sized constructed neighboring action space instead of exploring all the actions.
%can construct a smaller discrete neighboring action space using a continuous action and 

% \begin{table}[t]
% \centering
% \nblcaption{Comparison between the proposed method and one shot transmission.}
% \label{tab:one-shot}
% \begin{tabular}{|c|c|c|c|}
% \hline
% Method & \makecell{Packet Loss\\Rate} & \makecell{Weighted\\MSE} & \makecell{Transmission\\time} \\
% \hline
% \makecell{One shot\\transmission (1 W)} & 1.0 & 1.0 & -- \\
% \hline
% Proposed method (1 W) & \textbf{0.18} & \textbf{0.44} & 2.43 s \\
% \hline
% \makecell{One shot\\transmission (3 W)} & 0.92 & 0.92  & \textbf{1.38 s} \\
% \hline
% Proposed method (3 W) & \textbf{0.14} & \textbf{0.40} & 1.70 s \\
% \hline
% \end{tabular}
% \end{table}
\begin{table}[t]
\centering
\begingroup
\color{black}
\nblcaption{Comparison between the proposed method and one shot transmission.}
\label{tab:one-shot}
\begin{tabular}{|c|c|c|c|}
\hline
Method & \makecell{Packet Loss\\Rate} & \makecell{Weighted\\MSE} & \makecell{Transmission\\time} \\
\hline
\makecell{One shot\\transmission (1 W)} & 1.0 & 1.0 & -- \\
\hline
Proposed method (1 W) & \textbf{0.18} & \textbf{0.44} & 2.43 s \\
\hline
\makecell{One shot\\transmission (3 W)} & 0.92 & 0.92  & \textbf{1.38 s} \\
\hline
Proposed method (3 W) & \textbf{0.14} & \textbf{0.40} & 1.70 s \\
\hline
\end{tabular}
\endgroup
\vspace{-0.4cm}
\end{table}

In Fig.~\ref{fig:fb}, we show how the number of feedback messages transmitted by 8 users at each time slot varies in an episode.
% This metric allows us to infer whether users are coordinating their feedback transmissions. With effective cooperation, feedback to be transmitted is more evenly distributed across time slots rather than being concentrated in a few slots, thereby reducing transmission congestion and collision risks.
%We also show the variance of the number of feedback to represent its fluctuations.
%通过这个数据，我们可以看出用户之间是否有合作，使得每个time slot均匀分配feedback而不是集中在某几个time slot
From this figure, we see that the proposed method transmits around 5 feedback messages per time slot, and the variance of the number of transmitted feedback messages over 6 time slots is 0.14.
In contrast, the independent Q scheme unreliably transmits either 8 feedback messages or 3 feedback messages with a variance of 3.33, which may either increase feedback transmission delay or waste feedback transmission opportunity.
This is because the proposed method enables the users to collaboratively determine when to transmit their feedback messages thus optimizing the number of feedback messages that should be transmitted under a limited bandwidth channel and meeting the feedback transmission delay requirements. 
% In contrast, each user that uses the independent Q scheme only considers its individual reward and does not consider the actions of other users, causing uncoordinated feedback transmissions and a higher risk of feedback transmission failures. 
In contrast, under independent Q-learning, users transmit feedback independently without accounting for others' actions.

\nbl{Table~\ref{tab:one-shot} shows the gains (i.e., packet loss rate, weighted MSE, transmission time) achieved by the proposed method compared to the one shot transmission method where the transmitter directly transmits the entire image without receiving feedback under different downlink power levels.
From Table~\ref{tab:one-shot}, we see that the proposed method improves the packet loss rate by up to 82.1\% and the weighted MSE by up to 66.4\% compared to one shot transmission when the downlink power is 1~W.
These gains are achieved by the fact that the proposed method uses CSI and content preference feedback to optimize sub-image transmission and sub-channel allocation.
Moreover, since the packet loss rate of one shot transmission is 1 (i.e., all the image transmitted are lost), we cannot calculate its transmission time when the downlink power is 1~W.
Table~\ref{tab:one-shot} also shows that the one shot transmission method reduces the transmission time by up to 19.03\% compared to the proposed method when the downlink power is 3~W. This is because one shot transmission sends the whole image at once, while the proposed method transmits the image sequentially with feedback, which introduces additional transmission time for feedback and sequential sub-image transmission.
However, the proposed method can improve the packet loss rate and weighted MSE by up to 84.4\% and 55.9\% compared to one shot transmission.}
% This is because the proposed
% method can use the CSI and content-preference feedback to assist the BS in optimizing the
% sub-image transmission according to both channel quality and user-specific semantic importance.

In Fig.~\ref{fig:visualization}, we show the visualization of the sub-image and feedback transmission during an image transmission.
%1: Independent Q : Step2:CSI delay -> Step3:fail
%2: VDAC-DNC: Step6: bed fail -> not important and leave it to the poor channel
%3：VDAC-DNC: Step1: CSI feedback
% 由于系统设置每个三个time step CSI会改变, VDAC-DNC在step1和step4及时更新CSI确保sub-image能够可靠的传输，
% 在确保知道CSI后，发送相应的content feedback，确保重要的sub-image能够成功传输
% 同时在第一个step知道CSI和content weight之前，不会贸然传输过多sub-image
% 在最后一个step，考虑到bed的weight较低，将其安排在较差的信道
% 相反，Independent Q在不知道CSI时，BS不会考虑user的action，因此传了较多的sub-image，使得bed传输失败，同时user不考虑BS的动作，选择传输content feedback，使得后续传输仍然可能失败
% 而在Step2时，则由于多用户同时发送feedback，导致CSI feedback delay，间接使得Step3的book传输失败
% 且
% 因此整体接收到的image missing part比proposed method多很多。
% 在图中，BS发送时的×代表sub-image在传输时由于信道条件差导致丢包，而✔则代表sub-image成功传输
% 而在CSI feedback的1个✔代表该时刻有有一张sub-image在用户处被成功接收
% 而content feedback则会发送user处未收到的sub-image的且未发送过的最重要的4个content weight给BS，直到发送完所有可以发送的content weight（e.g.,In VDAC-DNC, 由于Step 5前所有能发的content weight已经发送完毕，此时没有多余的content可以发）
% 同时，由于在发送feedback时的宽带有限，如果在某一个时刻发送feedback的用户太多，则可能会到feedback传输被delay（e.g., in independent Q scheme，Step2 的CSI feedback delayed）
%在Step1的时候，由于BS没有任何信息，此时BS会随机选择sub-image进行发送，但proposed method会更加稳健，只传输一张确保其他可能重要的sub-image不会丢包
% In VDAC-DNC method，BS在Step2收到content feedback后，在Step3时，将重要的“book”和相对不那么重要的“shelf”一起发送，提高了sub-channel利用率和图像传输效率
In the figure, each $\times$ beside the sub-image implies that the sub-image is not received by the user due to poor channel conditions, whereas each $\checkmark$ represents that the sub-image is successfully delivered to the user. 
The CSI feedback with a $\checkmark$ at a given time step indicates that one sub-image is successfully received by the user at that step.
For content-preference feedback, each user sends the semantic weights of the top-$4$ most important sub-images that are (i) still missing at the user side and (ii) have not been sent before. This continues until all eligible content importance weights have been sent.
For example, at step 5 of the VDAC-DNC method, all content importance weights have been sent to the transmitter.
Hence, no additional content importance weights needs to transmit.
%(e.g., under VDAC-DNC, by step~5 all eligible content weights have already been transmitted, so no additional content feedback is available afterward).
In this simulation, CSI distributions are varied every three time slots.
From the figure, we see that the proposed method enables the BS to obtain the CSI at steps 1 and 4 to maintain a high data rate for sub-image transmission.
%the user to receive the most important sub-images and regenerate the image accurately.
% This is because the BS obtains the CSI at Steps 1 and 4 to maintain a high data rate for sub-image transmission.
% At Step~1, the BS has no CSI or content information, so it can only choose sub-images randomly for transmission. 
% The proposed method acts more conservatively compared to the independent Q scheme by sending only one sub-image to reduce the risk of losing other potentially important sub-images.
Fig.~\ref{fig:visualization} also shows that, at step~2, the CSI feedback in the independent Q scheme is 
not transmitted to the BS due to the delay constraint.
%delayed (i.e., does not meet the requirement of feedback transmission).
This is because the uplink bandwidth for feedback transmission is limited, and users in the independent Q scheme do not coordinate their feedback transmission, causing a large number of users to send feedback simultaneously.
% At Step~3, the user chooses not to send the feedback to the BS.
% This is because the channel may not be good enough to ensure the feedback can be sent as required, thus the user chooses not to waste power transmitting the feedback.
%too many users transmit feedback simultaneously, feedback will  be delayed. 
From Fig.~\ref{fig:visualization}, we also see that under the proposed VDAC-DNC, after the BS receives content-preference feedback at step~2, it schedules transmissions more efficiently at step~3 by sending the second important \emph{book} together with a less important \emph{shelf}.
%which improves sub-channel utilization and overall image transmission efficiency.

In Fig.~\ref{fig:dif-fb}, we show how the types of feedback transmitted by the users affect the accumulated reward.
In this figure, we compare the proposed VDAC-DNC method with: i) the proposed method in which the users only send content preference feedback to the BS, ii) the proposed method in which the users only send CSI feedback to the BS, and iii) the proposed method in which the users remain silent and will not send any feedback to the BS.
From Fig.~\ref{fig:dif-fb}, we see that the proposed VDAC-DNC method can improve the reward by up to 8.3\%, 7.1\%, and 13.1\% compared to the considered schemes (i), (ii), and (iii).
The 8.3\% gain stems from the fact that the proposed method can send CSI feedback to the BS, allowing the BS to learn the channel status thus improving sub-image transmission rate.
The 7.1\% gain stems from the content preference feedback transmission allowing the BS to prioritize important sub-images.
The 13.1\% gain is achieved by the joint CSI and content preference feedback transmission. 
Fig.~\ref{dif-range} shows how the accumulated reward changes as range $l$ of neighboring actions for DNC varies.
% the rewards of all considered schemes increase first and then decrease.
In this figure, we see that the gap in terms of the accumulated reward between the proposed method and the independent Q scheme increases as $l$ increases from 1 to 4.
%这是因为当l增大时，action set y会增大，而我们的算法由于critic由于训练时考虑全局信息，能在更大的action set选出更好的，所以相比independent Q能带来更低perfermance 下降
This is because, as $l$ increases, the neighboring action space increases.
The critic of the proposed method, which is trained using global information, can still select neighboring actions that increase the team reward, thus leading to a mild performance decrease.
In contrast, the independent Q scheme relies on a locally trained critic and selects the neighboring actions that increase the individual reward but may decrease the team reward. 
This problem becomes worse as the neighboring action space increases.
Fig.~\ref{dif-range} also shows that the proposed method with $l=2$ improves the accumulated reward by up to 4.79\% and 1.78\% compared to the designed method with $l=0$ and 4.
The 4.79\% gain stems from the fact that the designed method with $l=2$ can explore neighboring actions to help the BS find a better discrete action (i.e., a better sub-image allocation policy) than the initially chosen discrete action $\bar{\boldsymbol{a}}_{v,t}^{\textrm{B}}$. 
% In contrast, the designed method with $l=0$ cannot find any alternative neighboring action that is better than the action $\bar{\boldsymbol{a}}_{v,t}^{\textrm{B}}$.
The 1.78\% gain stems from the fact that when $l= 4$, the BS must traverse a large number of neighboring actions and it may select a neighboring action whose $Q$-value is overestimated by the critic networks, thus decreasing the accumulated reward.

\begin{figure}[t]
    \centering
    \includegraphics[width=8cm]{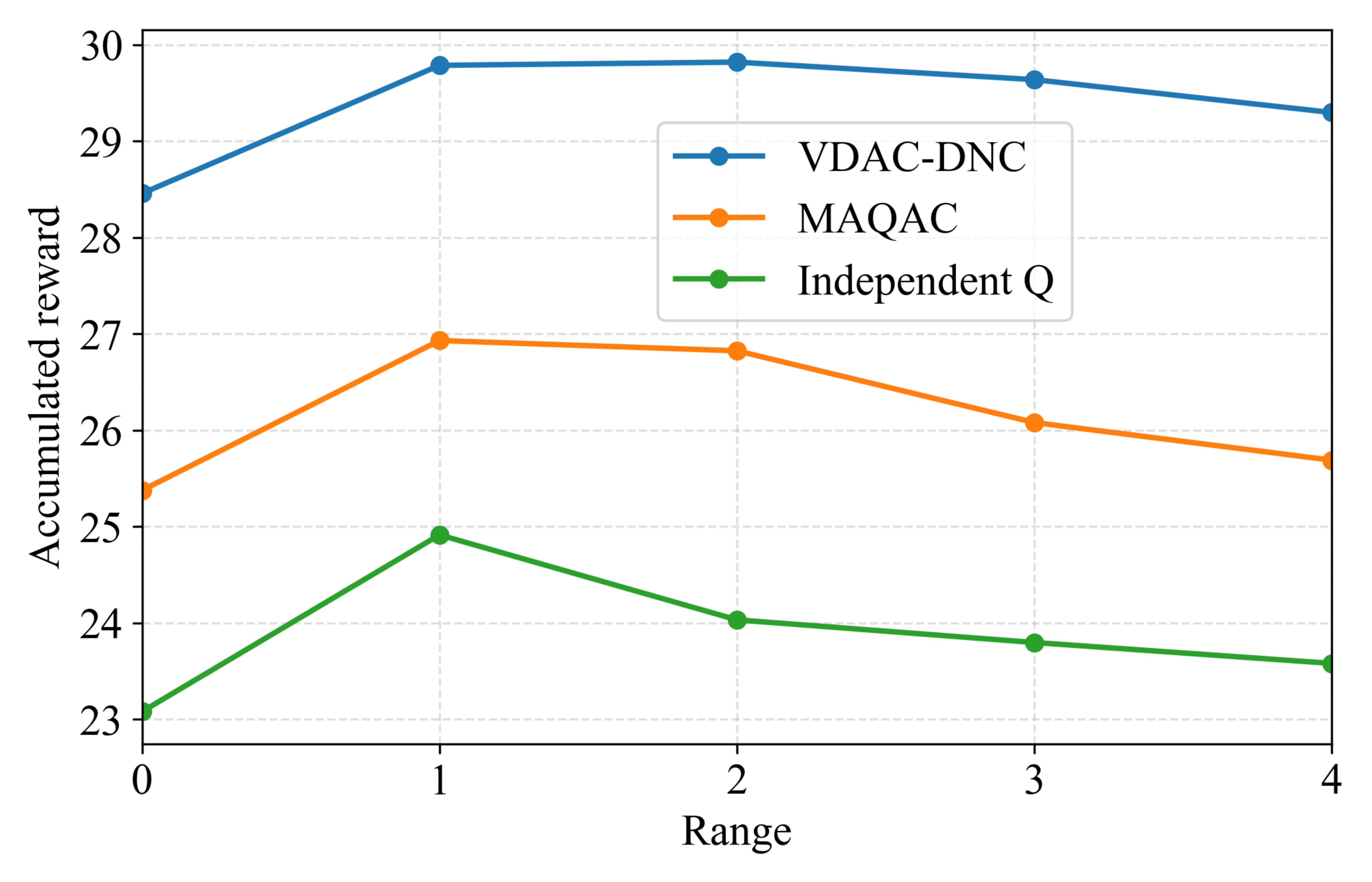}
    \vspace{-0.2cm}
    \caption{Accumulated reward versus the DNC search range}
    \label{dif-range}
    \vspace{-0.4cm}
\end{figure}
\begin{figure}
    \centering
    \includegraphics[width=8cm]{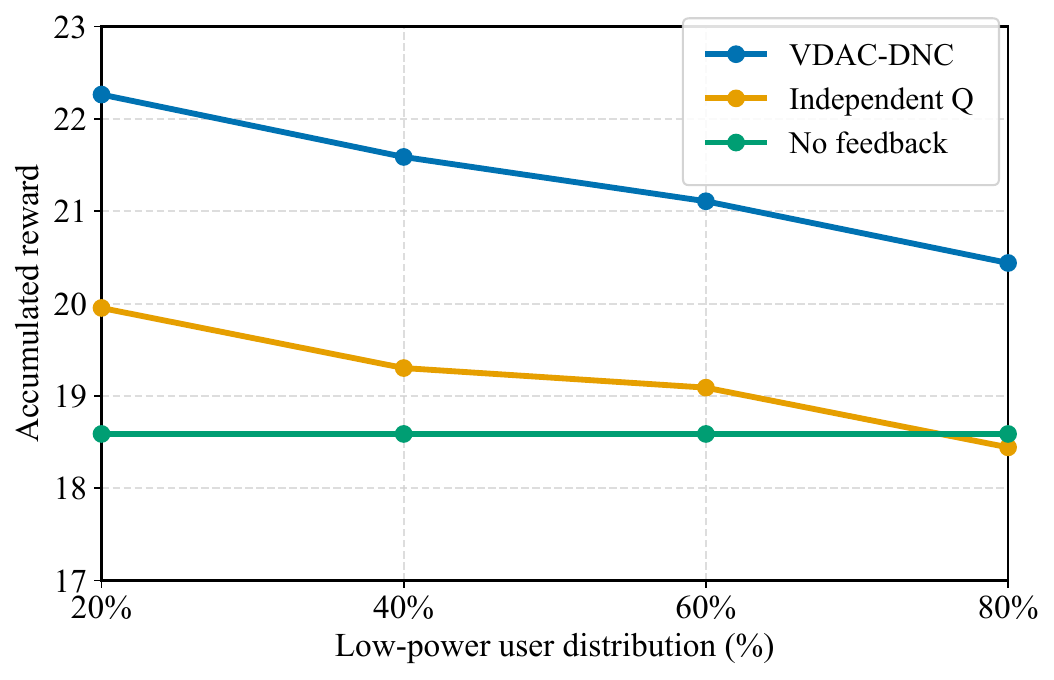}
    \vspace{-0.2cm}
    \nblcaption{The impact of device heterogeneity on the proposed framework}
    \label{fig:device_heterogeneity}
    \vspace{-0.4cm}
\end{figure}
\nbl{In Fig.~\ref{fig:device_heterogeneity}, we show how the accumulated reward is affected by the power budget of devices.
Here, some devices have higher feedback power budgets (e.g., $p^{\text{U}}_v = 1$~W and $p_{\text{max}} = 6$~W) and the remaining devices have lower feedback power budgets (e.g., $p^{\text{U}}_v = 0.001$~W and $p_{\text{max}} = 0.006$~W). In this figure, all users have the same CSI conditions such that we do not consider the impact of CSI. 
%and compare the reward performance of these two groups.
From Fig.~\ref{fig:device_heterogeneity}, we see that 
as the number of users with lower feedback power budgets increases from 20\% to 80\%, the reward of the proposed method increases by up to 11.84\%. This is because the devices with high power budgets can transmit feedback more frequently than the devices with low power budgets, thus providing more channel status and user preference information to assist the BS in optimizing the sub-image transmission.
Fig.~\ref{fig:device_heterogeneity} also shows that the reward gap between the proposed method and the independent Q method increases as the power budget of the devices increases. This is because our designed method allows users to coordinate their feedback transmission when users have sufficient power budgets.}

% \begin{figure}[t]
%     \centering
%     \includegraphics[width=8cm]{Journal_Figure/time_per_ep_vs_range_dm14.pdf}
%     \caption{Runtime versus the DNC search range}
%     \label{dif-range-time}
% \end{figure}

% \begin{figure}[t]
%     \centering
%     \includegraphics[width=8cm]{Journal_Figure/reward_vs_SA_dm14.pdf}
%     \caption{Reward versus the DNC SA steps}
%     \label{dif-SA}
% \end{figure}

% \begin{figure}[t]
%     \centering
%     \includegraphics[width=8cm]{Journal_Figure/convtime_vs_SA_dm14.pdf}
%     \caption{Convergence time versus the DNC SA steps}
%     \label{dif-SA-runtime}
% \end{figure}

% \begin{figure}[t]
%     \centering
%     \includegraphics[width=8cm]{Figures/feedback.pdf}
%     \caption{Reward versus maximum allowable number of feedback}
%     \label{feedback}
% \end{figure}
\section{Conclusion}\label{se:conclusion}
In this paper, we have proposed a novel channel and content preference
feedback enabled SC framework in which a BS transmits
the images to a set of users while users cooperatively
transmit the feedback to the BS for more efficient and
accurate image data transmission.
%allocate sub-channels and select feedback for image transmission .
This image transmission problem has been formulated as an optimization problem, whose goal is to minimize the semantic-weighted MSE between the original image and the regenerated image via optimizing sub-channel allocation of the BS, power allocation and feedback selection of users.
To address this problem, we have designed a VDAC-DNC scheme that introduces actor-critic to VDN, offering higher training efficiency compared to the standard VDN.
In addition, the proposed method with the DNC module is more effective in finding better actions in large discrete action spaces. 
In particular, the DNC uses a continuous action produced by a continuous policy, which is easier to optimize than a complex discrete policy over large discrete action spaces, to construct a small discrete neighboring action space. 
Then, DNC searches for an action with the maximum $Q$ value from the small neighboring action space, thus avoiding traversing the large discrete action space.
Simulation results have shown that the proposed VDAC-DNC scheme significantly outperforms the baselines with regard to the quality of reconstructed images and convergence time.

% \section*{Appendix}
% Appendix

\bibliographystyle{IEEEbib}
\bibliography{references1}
\end{document}